\documentclass[aos, preprint]{imsart}

\RequirePackage{amsthm,amsmath,amsfonts,amssymb}
\RequirePackage[authoryear]{natbib}
\RequirePackage[colorlinks,citecolor=blue,urlcolor=blue]{hyperref}
\RequirePackage{graphicx}
\usepackage{subcaption}
\usepackage{lscape}
\usepackage{rotating} 
\usepackage{cleveref}
\usepackage{arydshln}

\startlocaldefs
\theoremstyle{plain}

\newtheorem{theorem}{Theorem}[section]

\newtheorem{lemma}[theorem]{Lemma}
\theoremstyle{remark}
\newtheorem{definition}[theorem]{Definition}

\newcommand{\bm}[1]{\boldsymbol{#1}}

\newcommand{\Rb}{\mathbb{R}}

\def\ba#1\ea{\begin{align*}#1\end{align*}} 
\def\banum#1\eanum{\begin{align}#1\end{align}} 

\endlocaldefs

\begin{document}

\begin{frontmatter}
\title{Nonparametric function-on-scalar regression \\ using deep neural networks}
\runtitle{Nonparametric function-on-scalar regression using deep neural networks}

\begin{aug}
\author[A]{\fnms{Kazunori}~\snm{Takeshita}\ead[label=e1]{takeshita@sigmath.es.osaka-u.ac.jp}} \and
\author[A,B]{\fnms{Yoshikazu}~\snm{Terada}\ead[label=e2]{terada.yoshikazu.es@osaka-u.ac.jp}},
\address[A]{Graduate School of Engineering Science,
Osaka University\printead[presep={,\ }]{e1}}

\address[B]{Center for Advanced Integrated Intelligence Research,
RIKEN \printead[presep={,\ }]{e2}}
\end{aug}

\begin{abstract}
We focus on nonlinear Function-on-Scalar regression, 
where the predictors are scalar variables, and the responses are functional data.
Most existing studies approximate the hidden nonlinear relationships using linear combinations of basis functions, 
such as splines. 
However, in classical nonparametric regression, 
it is known that these approaches lack adaptivity, 
particularly when the true function exhibits high spatial inhomogeneity or anisotropic smoothness.
To capture the complex structure behind data adaptively, 
we propose a simple adaptive estimator based on a deep neural network model.
The proposed estimator is straightforward to implement using existing deep learning libraries, 
making it accessible for practical applications.
Moreover, we derive the convergence rates of the proposed estimator for the anisotropic Besov spaces, which consist of functions with varying smoothness across dimensions.
Our theoretical analysis shows that the proposed estimator mitigates the curse of dimensionality 
when the true function has high anisotropic smoothness, as shown in the classical nonparametric regression.
Numerical experiments demonstrate the superior adaptivity of the proposed estimator, 
outperforming existing methods across various challenging settings.
Moreover, the proposed method is applied to analyze ground reaction force data in the field of sports medicine, demonstrating more efficient estimation compared to existing approaches.
\end{abstract}


\end{frontmatter}

\section{Introduction}
\leavevmode

In recent years, advancements in measurement devices have made it common to encounter data, 
such as spectral data, 
that capture the continuous variations for each subject at several discrete time points.
Functional Data Analysis (FDA) is a methodology that treats such data as partial observations of latent curves recorded at discrete time points.
This perspective facilitates analysis even in cases where the observation times or the number of observations vary among subjects.
Numerous extensions of classical multivariate analysis methods into the FDA framework have been proposed. For a comprehensive introduction, we refer the reader to \cite{Ramsay2005functiondataanalysis} and \cite{Wang2016functionbook}.

Functional regression has gained attention as one of the most important areas in the FDA.
Regression models can be classified into three types based on the nature of the response and predictor variables \citep{Ramsay2005functiondataanalysis}: 
(i) the function-on-function (FOF) regression, 
where both response and predictors are functional; 
(ii) the scalar-on-function (SOF) regression, 
with the scalar response and functional predictors; and (iii) the function-on-scalar (FOS) regression, 
where the response is functional, and the predictors are scalar. 
This study focuses on the third type, the FOS regression.

Let $Y$ be the functional response on a time interval $I$ and $X_1,\ldots, X_d$ be the multiple scalar predictors. 
The FOS regression model is formulated as follows:
\begin{equation*}
    Y(t) = f^{\circ} (X_1, \ldots, X_d, t) + \xi (t), \quad t \in I,
\end{equation*}
where $f^{\circ}$ is the true regression function and $\xi$ is the random error function.

The FOS regression has been applied in various contexts within biological sciences. 
For example, it has been utilized in physical activity (PA) research 
\citep{goldsmith2016assessing, kowal2020bayesian, ghosal2023variable}, genome studies \citep{barber2017function, FAN2017167,parodi2018simultaneous}, and Alzheimer's disease research \citep{cai2021Robust}.
As an application of this paper, 
we will address a prediction task using ground reaction force (GRF) data, 
commonly used in sports medicine. 
The details of this application are described in Section \ref{sec : Application}.

Many studies have focused on linear models for predictor variables due to their statistical simplicity and interpretability (e.g., 
\citealp{faraway1997regression, Ramsay2005functiondataanalysis}). 
However, while the assumption of linearity might be restrictive in many real-world datasets, 
the development of nonlinear FOS regression models has been relatively underexplored.
\cite{scheipl2015functional} introduced an extensive framework for functional additive mixed models (FAMM).
The FAMM possesses a high level of flexibility, allowing for interactions or nested structures of predictors, and includes the varying coefficient single-index model (e.g., \citealp{luo2016single, li2017functional}) as a special case.
This approach is highly versatile for constructing basis functions on such a joint space.
However, the FAMM faces two main issues.
Firstly, an explicit specification of the model structure is required. 
Modern datasets often exhibit complex structures, making it difficult to define the model structure based on reasonable assumptions.
Second, relaxing the model assumptions to reduce the risk of misspecification can lead to an exponential increase in the number of parameters relative to the dimensionality, 
potentially reducing estimation efficiency.

To address these limitations, \cite{luo2023nonlinear} developed a more flexible method based on neural networks with a single hidden layer, demonstrating the universal approximation property of their model. 
However, this method presents two main challenges. 
First, since the model is essentially composed of a shallow network, it might struggle to efficiently obtain expressive power. 
In various settings, several studies have shown that shallow networks require significantly more parameters than deep networks to approximate functions (e.g., \citealp{cohen2016expressive, Daniely2017DepthSF, eldan2016power, mhaskar2016deep, safran2017depth}). 
For instance, \cite{safran2017depth} demonstrated that even for a simple function class, such as the indicator on a sphere, a 3-layer network can easily approximate it with a polynomial order of width relative to the input dimension, whereas a 2-layer network requires at least an exponential order of width.
Second, this issue is also common to the FAMM, as estimation based on basis functions lacks adaptivity, as will be discussed later.
The proposed estimation function includes integration, which makes precise modeling impractical. Therefore, the model is discretized in both the time and integration directions and represented using the cubic B-spline basis. 
Additionally, while Theorem 3 in \cite{luo2023nonlinear} derives the convergence rate of the estimator, the proof contains a critical flaw, 
leaving the correct convergence rate unestablished (see Appendix \ref{sec: Appendix D}).

Many existing methods rely on linear estimators with respect to the response variables, 
such as kernel estimators with fixed bandwidths and regression splines with equidistant knot points. 
In classical nonparametric regression, these estimators are non-adaptive and can only attain sub-optimal rates of convergence when the smoothness of the true function is highly spatially inhomogeneous (e.g., \citealp{donoho1998minimax, ZHANG2002256}). 
Although \cite{luo2023nonlinear} introduces a nonlinear estimator, 
it relies on splines with equidistant knot points, which is expected to inherit this limitation.

Deep neural networks (DNNs) offer a promising alternative to overcome these limitations. 
DNNs have achieved significant breakthroughs in solving complex problems that were previously intractable using traditional machine-learning methods. Their adaptivity and flexibility have driven remarkable advancements across diverse fields, including computer vision and speech recognition (e.g., \citealp{nassif2019speech, voulodimos2018deep}). 
Beyond these, DNN-based technologies have significantly influenced natural language processing (e.g., large language models like GPT) \citep{brown2020language}, as well as cutting-edge image synthesis and recognition tasks powered by diffusion models \citep{rombach2022high}.

In recent years, theoretical advancements in understanding the superior performance of deep neural networks (DNNs) have been made across various settings. Numerous studies have investigated the conditions under which DNNs outperform linear estimators (\citealp{imaizumi2019deep, schmidt2020Nonpara, suzuki2018adaptivity, suzuki2021deep, hayakawa2020minimax}). 
For instance, \cite{suzuki2021deep} demonstrated that DNNs achieve superior performance compared to linear estimators, such as kernel ridge regression and spline methods, when the true function is in the Besov space and has high spatial inhomogeneity. 
Similarly, \cite{schmidt2020Nonpara} showed that DNNs outperform wavelet-based methods when the true function can be represented as a composite function in the H\"{o}lder space.
Moreover, several studies have demonstrated that DNNs can avoid the curse of dimensionality in various settings. (e.g., \citealp{suzuki2018adaptivity, schmidt2019deep, schmidt2020Nonpara, imaizumi2019deep, bauer2019deep, chen2022nonparametric, suzuki2021deep}).

Building on these insights, we propose a simple nonlinear FOS regression method utilizing DNNs. 
The proposed method offers two key advantages.
First, our method naturally incorporates the architecture of existing deep neural networks into the FOS regression framework, enabling straightforward implementation.
Second, unlike traditional non-adaptive approaches, the proposed method automatically constructs adaptive estimators through standard stochastic optimization techniques, even in complex settings where the true function exhibits spatial inhomogeneity or anisotropic smoothness. 
To establish the theoretical properties of the proposed estimator, we assume that the true function is in the anisotropic Besov space (Details will be discussed in Section \ref{sec : Theoretical properties}).
Using the theoretical results of \cite{suzuki2021deep}, we demonstrate the convergence rate of the proposed estimator.
Our results demonstrate that similar to classical nonparametric regression, the proposed method effectively avoids the curse of dimensionality when the true function has anisotropic smoothness.
Furthermore, numerical experiments show that the proposed method outperforms existing approaches regarding efficiency and estimation accuracy under such scenarios.

This paper is organized as follows: Section \ref{sec : FOS-DNN model} presents the proposed nonparametric FOS regression model using deep learning. Section \ref{sec : Theoretical properties} introduces the anisotropic Besov space and derives the convergence rate of the proposed estimator. Sections \ref{sec : Simulation} and \ref{sec : Application} illustrate the performance of the proposed estimator compared to existing methods through simulation studies and real-data analysis, respectively.

\textbf{Notation} \par 
Let $\mathbb{Z}_+$ denote the set of non-negative integers. 
Write $\mathbb{Z}_+^d = \{(z_1, \ldots, z_d) \ | \ z_i \in \mathbb{Z}_+\}$, 
$\mathbb{R}_+ = \{x \geq 0 \ | \ x \in \mathbb{R}\}$, 
and $\mathbb{R}_{++} = \{x > 0 \ | \ x \in \mathbb{R}\}$.
Let $A^\top$ denote the transpose of the matrix $A$.
For $\boldsymbol{\alpha} = (\alpha_1, \ldots,\alpha_d)^{\top} \in \mathbb{R}^d\;(d\in \mathbb{N})$,
we define 
$|\boldsymbol{\alpha}| = \sum_{j=1}^D |\alpha_j|^2$.
For $x \in \mathbb{R}$, $\lfloor x \rfloor$ is the largest integer less than or equal to $x$,
 $\lceil x \rceil$ is the smallest integer greater than $x$ and $(x)_+$ returns the value of $x$ if $x$ is non-negative and returns $0$ if $x$ is negative.
For two sequences $(a_n)$ and $(b_n)$,
if there exists a constant $C$ such that $a_n \leq C b_n$ for all $n\in \mathbb{N}$,
we write $a_n \lesssim b_n$.
If $a_n \lesssim b_n$ and $b_n \lesssim a_n$, we write $a_n \approx b_n$.

Let $\mathcal{D}$ be a domain of the functions. For a measurable function $f : \mathcal{D} \rightarrow \mathbb{R}$,
define the $L_p$- and $L_\infty$-norms as 
$
    \|f\|_{L^p(\mathcal{D})} = \left(\int_{\mathcal{D}} |f(\bm{x})|^p\,d\bm{x}\right)^{1/p} \;(0 < p < \infty)
$
and 
$
\|f\|_{L^\infty (\mathcal{D})} = \sup_{\bm{x} \in \mathcal{D}} |f(\bm{x})|.
$
For a probability density function $p_{\boldsymbol{X}}$ on $\mathcal{D}$,
we define 
$
    \|f\|_{L^p(p_{\boldsymbol{X}}(\mathcal{D}))} = (\int_{\mathcal{D}} |f(\bm{x})|^p p_{\boldsymbol{X}}(\bm{x})\,d\bm{x} )^{1/p}
$ $(0 < p < \infty)$, and 
$
\|f\|_{L^{\infty}(p_{\boldsymbol{X}}(\mathcal{D}))} =\sup_{\boldsymbol{x} \in \mathcal{D}} |f(\bm{x})p_{\boldsymbol{X}}(\bm{x})|.
$
\section{FOS-DNN model} \label{sec : FOS-DNN model}
\leavevmode

Firstly, we describe the problem setting in this work.
Write $D = d+1$ and $\Omega = [0,1]^{D}$.
We consider the following nonparametric FOS regression model:
\begin{equation}
    Y_i(t) = f^{\circ}(X_{i1}, \ldots, X_{id}, t) + \xi_i(t) \quad(0 \leq t \leq 1), \label{eq:nonparaFOS}
\end{equation}
where $f^\circ : \Omega \rightarrow \mathbb{R}$ is the unknown true function, $\boldsymbol{X}_i = (X_{i1}, \ldots, X_{id})^\top$ is generated from a probability distribution $P_{\boldsymbol{X}}$ on $[0,1]^d$, 
and $\xi_i$ is a continuous
sub-Gaussian process on $[0,1]$ with mean $0$ and variance proxy $\sigma^2$. 
More precisely, we assume that the process $\xi_i$ is continuous and
that 
$
\mathbb{E}\left[\exp\left(s \int_0^1 \xi_i(t) f(t) \, dt\right)\right] \leq \exp(s^2 \sigma^2 / 2)
$
for any $f \in L^2([0,1])$ with $\|f\|_{L^2([0,1])} = 1$ and for any $s \in \mathbb{R}$.
We assume that the training data $D_n = \{(\boldsymbol{X}_1, Y_1), \ldots, (\boldsymbol{X}_n, Y_n)\}$ are independently and identically distributed. 
The goal of the regression problem is to estimate the unknown true function $f^{\circ}$ based on the observed data $D_n$.

Now, we propose an estimation method for the true function $f^{\circ}$ via deep learning, referred to as the \textit{FOS-DNN} model. 
Specifically, we estimate the true function $f^{\circ} : [0,1]^d \times [0,1] \rightarrow \mathbb{R}$ using a feedforward neural network architecture with predictors $\bm{X}$ and time $t$ as the input (see Figure \ref{fig:proposed method}).
To define the FOS-DNN model mathematically, we specify an activation function $\eta : \mathbb{R} \rightarrow \mathbb{R}$. 
In this study,  we use the ReLU activation function, defined as $\eta (x) = \max \{x, 0\}$. 
For $\boldsymbol{x}=(x_1,\dots,x_p)^\top\in \Rb^p$, the activation function $\eta$ is applied element-wise, 
that is, $\eta(\boldsymbol{x}) = (\eta(x_1),\dots,\eta(x_p))^\top$.
Let $\bm{z} = (x_1, \ldots, x_d, t)^\top$ be the input, and let $L$ and $W$ denote the depth and width of the neural network, respectively. 
Additionally, we denote the width vector $\bm{c} = (c_1, \ldots, c_L) \in \mathbb{N}^{L}, $ where $c_1 = D, c_L = 1 \;\text{and}\; c_l = W \;\text{for}\; 1 < l < L$.
The FOS-DNN model, represented as $\Phi(L, W)$, is defined as any function of the following form:
\begin{equation}
     f : \mathbb{R}^{D} \rightarrow \mathbb{R}, \; \bm{z} \mapsto f(\bm{z}) = (\boldsymbol{\mathcal{W}}^{(L)} \eta(\cdot) + \boldsymbol{b}^{(L)}) \circ \cdots \circ(\boldsymbol{\mathcal{W}}^{(1)} \boldsymbol{z} + \boldsymbol{b}^{(1)}), \label{eq:FOS-DNN nottheo}
\end{equation}
where $\boldsymbol{\mathcal{W}}^{(l)} \in \mathbb{R}^{c_{l+1} \times c_l}$ is a weight matrix, and $\bm{b}^{(l)} \in \mathbb{R}^{c_l}$ is a bias vector. 
The set of network parameters $\Theta_{L,W}$ is defined as
\begin{equation*}
    \Theta_{L,W} = \left\lbrace \left( \text{vec}(\bm{\mathcal{W}}^{(1)})^\top, \bm{b}_1^\top, \ldots, \text{vec}(\bm{\mathcal{W}}^{(L)})^\top,  \bm{b}_L^\top \right)^\top : \boldsymbol{\mathcal{W}}^{(i)} \in \mathbb{R}^{c_{l+1} \times c_l},  \bm{b}_l \in \mathbb{R}^{c_l}, \ 1 \leq l \leq L \right\rbrace.
\end{equation*}
where $\text{vec}(\cdot)$ is the vectorization operator which stacks rows of a matrix $\bm{\mathcal
{W}} \in \mathbb{R}^{p \times q}$ into a column vector $\bm{\omega} \in \mathbb{R}^{pq}$.    

\begin{figure}
    \centering
    \includegraphics[width=0.5\linewidth]{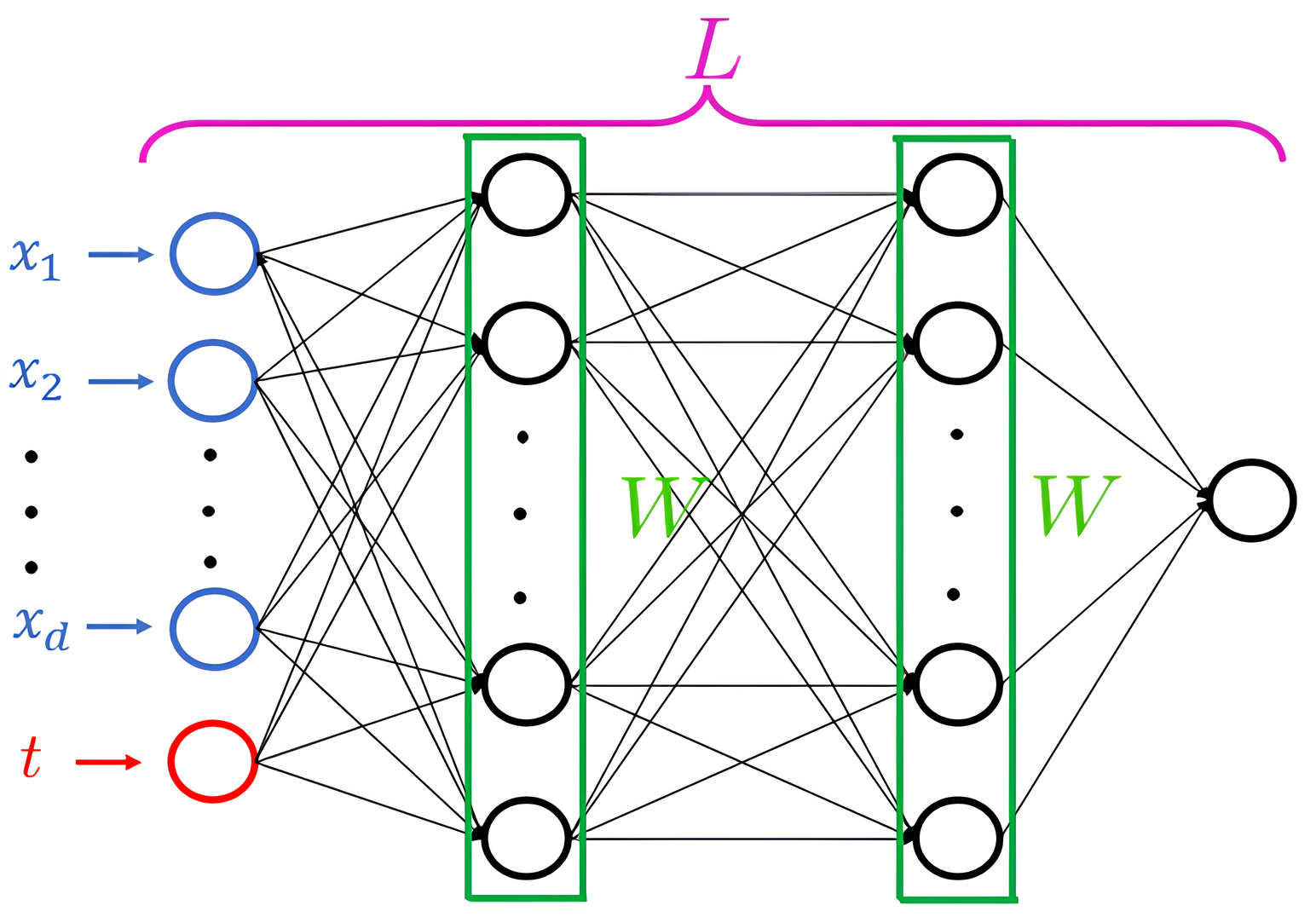}
    \caption{Illustration of \textit{FOS-DNN} model. The model takes predictor variables $\boldsymbol{X}$ and $t$ as inputs and estimates the true function $f^\circ$  using a feedforward deep neural network with specified depth $L$ and width $W$.}
    \label{fig:proposed method}
\end{figure}

Next, we outline the estimation procedure of FOS-DNN model. 
The depth $L$ and width $W$ of the network are fixed prior to estimation. 
Although computationally expensive, they can be selected using methods such as cross-validation.
Given a training data $D_n = \{(\boldsymbol{X}_i, Y_i)\}_{i=1}^n$, 
the network parameters are estimated by solving the following least-squares optimization problem:
\begin{align}
    \hat{\theta}_{L,W} 
    = \underset{\theta \in \Theta_{L,W}}{\operatorname{argmin}}  \frac{1}{n} \sum_{i=1}^n \int_0^1 \{Y_i(t) - f_\theta(\boldsymbol{X}_i, t)\}^2 dt. \label{eq:objectivefunction natural}
\end{align}
where $f_\theta$ is an element of $\Phi(L,W)$ corresponding to the network parameter $\theta$.

However, in practice, the response variable $\{Y_i(t)\}_{0 \leq t \leq 1}$ is not continuously observed. 
In this study, we consider sufficiently dense discrete observations where the observation points may vary among subjects.
This setting is more general than the fixed observation point in many existing studies.
Specifically, we assume that $\{Y_i(t)\}_{0 \leq t \leq 1}$ for $i = 1, \ldots, n$ is observed discretely at $\boldsymbol{t}_i = (t_{i1}, \ldots, t_{iN_i})^\top$ with $0 = t_{i0} \leq t_{i1} < \cdots < t_{iN_i} \leq 1$. Under $N_i$ is sufficiently large for each $i$, the loss function in (\ref{eq:objectivefunction natural}) can be approximated as follows:
\begin{align}
    \hat{\theta}_{L,W} &\simeq \underset{\theta \in \Theta_{L,W}}{\operatorname{argmin}} \   \frac{1}{n} \sum_{i=1}^n \sum_{i1 \leq n_i \leq iN_i}\{Y_i(t_{n_i}) - f_\theta(\boldsymbol{X}_i, t_{n_i})\}^2 (t_{n_i} - t_{n_{i-1}}). \label{eq:approxobj} 
\end{align}
In practice, we use stochastic gradient descent methods to obtain the solution of the optimization problem (\ref{eq:approxobj}). 
Due to the highly non-convex nature of this optimization problem, 
many solutions may become trapped in local minima. 
However, many studies have shown that the local minima obtained by large-size networks exhibit good performance.
For example, \cite{choromanska2015loss} showed that many local minima are located within a specific band whose lower bound corresponds to the loss at the global minimum and that the number of local minima outside this band decreases exponentially as the size of the network increases. 
Additionally, \cite{du2019gradient} proved that gradient descent converges to the global minimum when the network size is sufficiently large.

When the training data size is small, the minimizer obtained from the training data may lead to overfitting.
Regularization is a widely recognized and effective theoretical and practical method to address this problem. Explicit techniques, such as $L_1$- or $L_2$-regularization applied directly to the loss function (\ref{eq:approxobj}), and implicit approaches, such as batch normalization and dropout integrated into the neural network architecture, are commonly used and can also be applied to the FOS-DNN model.

\section{Theoretical properties of the FOS-DNN estimator} \label{sec : Theoretical properties}
\leavevmode

In this section, we derive the specific convergence rates of FOS-DNN estimator. 
In the theoretical part of this study, we assume that all parameters of the FOS-DNN are bounded by a positive constant $B$. Additionally, we impose a constraint that all parameters' total number of nonzero components does not exceed $S$.
The class of FOS-DNN models with the above constraints is defined as
\begin{align*}
     \Phi(L,W,S,B) 
     := \left\lbrace f \; \text{of the form (\ref{eq:FOS-DNN nottheo})} : \sum_{l=1}^L (\|\boldsymbol{\mathcal{W}}^{(l)}\|_0 + \|\boldsymbol{b}^{(l)}\|_0) \leq S,  \max_{l} \|\boldsymbol{\mathcal{W}}^{(l)}\|_\infty \lor \|\boldsymbol{b}^{(l)}\|_\infty \leq B \right\rbrace,
\end{align*}
where $\|\cdot\|_0$ is the total number of non-zero components of the matrix, and $\|\cdot\|_\infty$ is the maximum of the absolute values of the components of the matrix. 
For the sake of simplicity, we denote $\Phi(L, W, S, B)$ briefly by $\Phi$.\par

Next, we construct a FOS-DNN estimator for the true function $f^\circ$. 
The theory of the FDA generally assumes that functional data are either ideally observed continuously or observed discretely.
For discrete observations, previous studies have shown that the balance between the number of observation points and the sample size is closely related to the asymptotic properties for estimating mean functions and the covariances functions, which are fundamental problems in the FDA (e.g., \citealp{zhang2016sparse}).
Discrete observations reflect realistic problem settings and are crucial for a detailed analysis of estimator properties. 
However, they make theoretical analysis difficult due to their complexity.
Therefore, in the theoretical part of this study, for simplicity, we assume that the response variables $\{Y_i\}_{i=1}^n$ are continuously observed on $[0,1]$. Additionally, suppose that a global minimizer within $\Phi$ is obtained exactly, without any approximation errors, as follows:
\begin{align*}
    \hat{f}_n = \underset{\bar{f} \in \bar{\Phi}_F}{\operatorname{argmin}} \frac{1}{n} \sum_{i=1}^n \|Y_i - \bar{f}(\boldsymbol{X}_i, \cdot)\|_{L^2([0,1])}^2
    =  \underset{\bar{f} \in \bar{\Phi}_F}{\operatorname{argmin}}  \frac{1}{n} \sum_{i=1}^n \int_0^1 \{Y_i(t) - \bar{f}(\boldsymbol{X}_i, t)\}^2\,dt,
\end{align*}
where $\bar{\Phi}_F$ is the \textit{clipping} of $\Phi$ defined by $\bar{\Phi}_F = \{\bar{f} = \min \{\max \{f, -F\}, F\} : f \in \Phi\}$ for a constant $F > 0$, which is attained by ReLU unit. The clipping is necessary to ensure the boundedness of the estimator.

The properties of the class of true functions play a crucial role in determining the convergence rates of the estimator. 
We assume that the true function $f^{\circ}$ is in an anisotropic Besov space, which is introduced by \cite{suzuki2021deep}.
This space includes functions with direction-dependent smoothness and generalizes well-known function spaces such as the H\"{o}lder, Sobolev, and Besov spaces. 
The flexibility of the anisotropic Besov space allows for a more realistic and adaptable framework, 
particularly in high-dimensional settings. 
In practical situations, it is common that only a subset of predictor variables has a significant impact on the response variable, 
exhibiting substantial variation, including discontinuities, while most others remain relatively smooth.

For a function $f \in L^p(\Omega)\;(p \in (0, \infty])$, the $r$\textit{-th modulus of smoothness} of $f$ is defined by 
    \begin{equation*}
        \omega_{r, p}(f, \boldsymbol{s}) := \sup_{\boldsymbol{h} \in \mathbb{R}^D : |h_i| \leq s_i} \|\Delta_{\boldsymbol{h}}^r (f)\|_p, \quad \text{for} \ \boldsymbol{s} = (s_1, \ldots, s_d, s_{D}) \in \Rb_{++}^{D},
    \end{equation*}
where
    \begin{equation*}
        \Delta_{\boldsymbol{h}}^r(f)(\boldsymbol{x}) := 
         \begin{cases}
            \displaystyle \sum_{j=0}^r \binom{n}{k} (-1)^{r-j} f(\boldsymbol{x} + j\boldsymbol{h}) & \text{if }\boldsymbol{x} \in \Omega, \boldsymbol{x} + r\boldsymbol{h} \in \Omega,\\
            0       & \text{otherwise}.
        \end{cases} 
    \end{equation*}
For a given $0 < p, q \leq \infty, \ \boldsymbol{\beta} = (\beta_1, \ldots, \beta_d, \beta_D)^\top, \ r = \max_i \lfloor \beta_i \rfloor$, let the seminorm $|\cdot|_{B_{p,q}^{\boldsymbol{\beta}}}$ be 
    \begin{equation*}
            |f|_{B_{p,q}^\beta} := 
            \begin{cases}
                \displaystyle \left( \sum_{k=0}^\infty [2^k \omega_{r,p}(f, (2^{- k/\beta_1}, \ldots, 2^{-k/\beta_d}, 2^{-k/\beta_D}))]^q \right)^{1/q} & (q < \infty) \\
                \sup_{k \geq 0} 2^k \omega_{r,p}(f, (2^{- k/\beta_1}, \ldots, 2^{-k/\beta_d}, 2^{-k/\beta_D}))) & (q=\infty)
            \end{cases}
            .
    \end{equation*}
    The norm of the anisotropic Besov space $B_{p,q}^{\boldsymbol{\beta}} (\Omega)$ is defined by $\|f\|_{B_{p,q}^{\boldsymbol{\beta}}} : = \|f\|_p + |f|_{B_{p,q}^{\boldsymbol{\beta}}}$, and $B_{p,q}^{\boldsymbol{\beta}} := \{f \in L^p(\Omega) \ | \ \|f\|_{B_{p,q}^{\boldsymbol{\beta}}} < \infty\}$. \par
    
    Additionally, for a $\boldsymbol{\beta} = (\beta_1, \ldots, \beta_d, \beta_D)^{\top} \in \mathbb{R}_{++}^D$, we define
\begin{equation*}
    \underline{\beta} := \min_{1\le i \le D} \beta_i,  \quad \overline{\beta} := \max_{1\le i \le D} \beta_i, \quad \text{and} \quad \tilde{\beta} := \left(\sum_{i=1}^D 1/\beta_j \right)^{-1}.
\end{equation*}
The parameter $\boldsymbol{\beta}$ represents the smoothness in each coordinate direction. 
When $\beta_i$ is large, the function is smooth in the $i$-th coordinate direction.
In this study, we denote $\beta_1, \ldots, \beta_d$ as the smoothness parameters corresponding to the predictors $X_1, \ldots, X_d$, 
and $\beta_D$ as the smoothness parameter for time $t$. 
The parameter $p$ controls the \textit{spatial inhomogeneity} of the smoothness, 
where a smaller value of $p$ implies spatially inhomogeneous smoothness, leading to features such as spikes and jumps. 
If the smoothness parameter satisfies $\beta_i > D/p$, the function is continuous in the direction of the $i$-th axis. However, if $\beta_i < D/p$, it is no longer continuous in that direction\footnote{\cite{luo2023nonlinear} assume the continuity of the true function. Therefore, the setting of our study is broader.}. The connections between Besov spaces and other key function spaces, such as H\"{o}lder spaces, are discussed in \cite{triebel2011entropy}.

The performance of $\hat{f}_n$ is evaluated using the mean squared estimation error
\[
E_{D_n} \left[\int_0^1 \|f^{\circ}(\cdot, t) - \hat{f}_n(\cdot, t)\|_{L^2(p_{\boldsymbol{X}}([0,1]^d))}^2 dt\right],
\]
where $E_{D_n}[\cdot]$ denotes the expectation with respect to the training data $D_n$. 
The following theorem provides an upper bound on the mean squared error between the true and estimator for the nonparametric FOS regression with deep learning. The proof is provided in Appendix \ref{sec : Appendix A}. 

\begin{theorem} 
\label{thm1}
    Suppose that $0< p, q\leq \infty$ and that $\boldsymbol{\beta} \in \mathbb{R}_{++}^D$ satisfy $\tilde{\beta} > (1/p - 1/2)_+$.
    Let $m \in \mathbb{N}$ satisfies $0 < \overline{\beta} < \min \{m, m-1+1/p\}$. 
    We define $\delta := (1/p - 1/2)_+, \ \nu:= (\tilde{\beta} - \delta) / (2\delta),\  W_o(D) := 6Dm(m+2) + 2D$. 
    For $N \in \mathbb{N}$, let
    \begin{align*}
        &L_1(D):= 3 + 2 \left\lceil \log_2 \left(\frac{3^{D \lor m}}{\epsilon c_{(D,m)}} \right) + 5 \right\rceil \lceil \log_2 (D \lor m) \rceil, \quad W_1(D):= NW_0, \\
        &S_1(D) := [(L-1)W_0^2 + 1]N, \quad B_1(D) := O(N^{D(1+\nu^{-1})(1/p - \tilde{\beta})_+}),
    \end{align*}
    where $\epsilon = N^{- \tilde{\beta}} (\log (N))^{-1}$, and $c_{(D,m)}$ is a constant depending only on $D$ and $m$.\par 
    Assume that the marginal distribution $P_{\boldsymbol{X}}$ has a density $p_{\boldsymbol{X}}$, 
    and there is a constant $R > 0$ such that $\|p_{\boldsymbol{X}}\|_\infty \leq R$. 
    Suppose that $f^{\circ} \in B_{p,q}^{\boldsymbol{\beta}}(\Omega) \cap L^\infty(\Omega)$ and that $\|f^{\circ}\|_\infty \leq F$ for $F \geq 1$.
    Then, letting $(W,L,S,B) = (L_1(D), W_1(D), S_1(D), (D + 1)B_1(D))$ with $N \approx n^{1/(2 \tilde{\beta} + 1)}$, 
    we have
    \begin{equation*}
            E_{D_n}\left\lbrack \int_0^1 \|f^{\circ}(\cdot, t) - \hat{f}_n(\cdot, t)\|_{L^2(p_{\boldsymbol{X}}([0,1]^d))}^2 dt \right\rbrack \lesssim n^{- 2 \tilde{\beta}/(2 \tilde{\beta} + 1)} (\log n)^3.
        \end{equation*}
\end{theorem}

The convergence rates in \Cref{thm1} are consistent with those in \cite{suzuki2021deep}.
Thus, \Cref{thm1} can be considered as a straightforward extension for the FOS regression setting.
According to \Cref{thm1}, the convergence rates improve as the smoothness of the true function $f^\circ$, which is represented by $\tilde{\beta}$, increases. 
Remarkably, the convergence rates do not directly depend on the dimensionality $d$ of the predictor variables $\boldsymbol{X}$. 
This result suggests that the proposed method may overcome the curse of dimensionality in certain cases.

When the true function is in the isotropic Besov space with uniform smoothness, 
such that $\beta_1 = \cdots = \beta_d = \beta_D = \underline{\beta}$, 
the exponent of $n$ in the convergence rate is given by
\begin{equation*}
    - \frac{2 \underline{\beta}}{2 \underline{\beta} + d + 1}.
\end{equation*}
The dimensionality $d$ appears in the exponent, resulting in the curse of dimensionality. 

Conversely, we consider the case where the true function is included in the anisotropic Besov space, particularly when only a subset of predictor variables exhibits significant variation with respect to the response variable. 
Specifically, we assume that $\beta_1 = \underline{\beta}$, while $\beta_2 = \cdots = \beta_d = \beta_D = \overline{\beta}$, with $\overline{\beta} = \alpha \underline{\beta}$ for some $\alpha > 1$.
Then, the exponent of $n$ in the convergence rate is given by
\begin{equation*}
    - \frac{2 \underline{\beta}}{2 \underline{\beta} + d / \alpha + 1}.
\end{equation*}
In this case, the influence of dimensionality $d$ is reduced by a factor of $1/\alpha$. This implies that the FOS-DNN method could avoid the curse of dimensionality when the smoothness of the true function is highly anisotropic. 
In high-dimensional settings, 
not all variables affect the output, and thus
it is reasonable for the true function to be not sensitive to perturbations along many input directions.

In the classical nonparametric regression setting, 
\cite{suzuki2021deep} demonstrated that the convergence rate derived in \Cref{thm1} is minimax optimal. 
Notably, when the spatial inhomogeneity is high ($p<2$), 
deep learning outperforms any linear estimators in the sense of the minimax optimality. 
This suggests that deep learning can adaptively detect both smooth and non-smooth regions of the true function from the data, enabling efficient estimation. 
Conversely, linear estimators using basis functions with a fixed spatial scale cannot adjust the resolution of the basis functions based on the data, making it difficult to capture local variations in smoothness effectively and thus preventing them from achieving the minimax optimal rate.

\section{Simulation} \label{sec : Simulation}
\leavevmode

In this section, we evaluate the performance of \textit{FOS-DNN} model through simulations based on various nonlinear FOS models. 
We compare our method with the \textit{FUA} \citep{luo2023nonlinear}, the FAMM \citep{scheipl2015functional}, and 
the linear FOS method \textit{fosr} from the R package \texttt{refund}.
For the FAMM, we focus on two specific forms, 
considering the practical difficulties in specifying the model precisely.
The first is the simplest additive model, referred to as \textit{famm.ad}, 
which estimates the true function in the form $f(\bm{X},t) = \sum_{j=1}^d f_j (X_j, t)$.
The second is the most flexible model, referred to as \textit{famm.nl}, 
which directly estimates the true function without any specific structure.
However, as validated in \cite{luo2023nonlinear}, 
the \textit{famm.nl} is computationally expensive and therefore is applicable only when $d$ is small.

We consider four simulation scenarios. 
The first two scenarios replicate the settings in \cite{luo2023nonlinear}, with $d=3$ and $d=5$, referred to as Scenario 1 and Scenario 1-A, respectively, 
where the true functions exhibit relatively low spatial inhomogeneity and anisotropy. 
Due to space limitations, Scenario 1-A is provided in Appendix \ref{sec: Appendix C}.
In contrast, the remaining two scenarios are designed to evaluate the performance of the proposed method under more challenging conditions, with $d=5$ and $d=10$, referred to as Scenario 2 and Scenario 3, respectively, where the true functions exhibit high spatial inhomogeneity and anisotropy. 
For scenarios with $d=5$ and $d=10$, the \textit{famm.nl} is excluded. 

For each scenario, we repeat the following procedure. 
In each replication, a training dataset and a test dataset are generated. While the size of the training dataset varies across scenarios, 
the sample sizes of the test datasets are fixed at $N_{\text{test}} = 1000$. 
Each sample response curve $\{Y(t)\}_{0 \leq t \leq 1}$ is observed at 100 equally spaced points $0 = t_1 < t_2 < \cdots < t_{100} = 1$. 
We then apply the estimated models to the test dataset to calculate the mean integrated squared prediction error (MISPE): 
\begin{equation*}
    \text{MISPE} = 
     \frac{1}{N_{\text{test}}} \sum_{i=1}^{N_{\text{test}}} \sum_{j=1}^{100}
     \{Y_i^{\text{test}}(t_j) - \hat{f}_i(t_j)\}^2 \Delta_t,
\end{equation*}
where $Y_i^{\text{test}}(t)$ denotes the $i$th sample response curve in the test data, $\hat{f}_i(t)$ represents the corresponding predicted curve and $\Delta_t = 0.01$.

In each scenario, 
we consider several candidates of network parameters for the \textit{FOS-DNN}. 
In Scenario 1, 
we examine combinations of the neural network widths $W=8,16,32$ and depths $L=5,6,7$. 
In Scenarios 1-A and 2, the width is chosen from $W=16,32,64$, and the depth is selected from $L=6,7,8$. 
In Scenario 3, the candidate configurations are identical to those in Scenario 2.
Additionally, we introduce an $L^2$-regularization parameter $\alpha$ to mitigate the overfitting, 
selecting it from $\alpha=10^{-9}, 10^{-7}, 10^{-5}, 10^{-3}, 10^{-1}$. 
In each setting, 
the tuning parameters $(W, L,\alpha)$ of the proposed method are determined using 3-fold cross-validation 
on a single training dataset, and the selected parameters are applied uniformly across all datasets within the scenario. 
In contrast, we select smoothing parameters of other methods via cross-validation on the training dataset for each replication. 

\subsection{Scenario 1} \label{subsec : Simulation 1}

In Scenario 1, we adopt the same setting as described in \cite{luo2023nonlinear}. 
We have three scalar predictors $(X_1, X_2, X_3)$ and analyze three different types of joint distributions for these predictors: (Xtype-1) $(X_1, X_2, X_3) \sim \text{Unif}([-1, 1]^3)$, 
where $\text{Unif}([-1, 1]^3)$ denotes the uniform distribution on $[-1,1]^3$; (Xtype-2) $(X_1, X_2, X_3) \sim \mathcal{N}_3(\boldsymbol{0}, \Sigma_{0.1})$; (Xtype-3) $(X_1, X_2, X_3) \sim \mathcal{N}_3(\boldsymbol{0}, \Sigma_{0.5})$, where $\mathcal{N}_d(\boldsymbol{0}, \Sigma)$ represents a $d$-dimensional multivariate normal distribution with mean zero and covariance matrix $\Sigma$, and $\Sigma_{\rho}$ is the covariance matrix with diagonal elements equal to one and off-diagonal elements equal to $\rho$.

We consider the following three different nonlinear models for the true function $f^\circ$:
\begin{align*}
    \textbf{Model 1}: \ f^\circ(X_1,X_2,X_3, t) &= \{X_1 \sin(4t)\}/(X_1^2 + 1) + e^{- \frac{(X_2 - 2)^2}{2}} + \{1 + t \cos(X_3)\}, \\
    \textbf{Model 2}: \ f^\circ(X_1,X_2,X_3, t) &= \{X_1 X_2 \cos(2 \pi t) + \log(1 + X^2 + X_3^2 + t^2)\} e^{-(X_1^2 + X_2^2 + X_3^2)/10}, \\
    \textbf{Model 3}: \ f^\circ(X_1,X_2,X_3, t) &= (X_1 t^2 + X_2 t - X_3)^2/\{1 + X_1^2 t + X_2^2 \sin^2(t) + X_3^2\}.
\end{align*}
Model 1 is a nonlinear additive model, while the others models have more complex nonlinear structure. 
In total, we consider nine settings based on different combinations of models and joint distributions for the predictors. 
In each setting, 
the predictors are generated, 
followed by the generation of $Y(t)$ according to the model $Y(t) = c f^\circ(X_1, X_2, X_3, t) + \epsilon(t)$. 
Here, $c$ is a constant chosen such that the integrated variance of $S(t) = c f^\circ(X_1, X_2, X_3, t)$ equals one, 
and the error term $\epsilon(t)$ is independently distributed as $\mathcal{N}(0, 0.1^2)$ for each $t \in [0, 1]$. 
Sample response curves for each setting are shown in Figure \ref{fig:Simulation 1} of Appendix \ref{sec: Appendix C}.

Table \ref{tab:1} presents the averages and standard deviations of the MISPEs over 100 replications for all settings with a training data size of $N_{\text{train}} = 200$ (for the results when $N_{\text{train}} = 1000$, Table \ref{tab:2} of Appendix \ref{sec: Appendix C}). 
For the \textit{FOS-DNN}, we set $W=32$, $L=6$, and $\alpha = 10^{-3}$. 
Additionally, for the \textit{FUA}, \textit{famm.nl}, 
and \textit{famm.ad}, the parameter settings in \cite{luo2023nonlinear} are used.
For all three distributions of $(X_1, X_2, X_3)$ in Model 1, 
\textit{FUA} and \textit{famm.ad} achieve lower prediction errors compared to the other methods, 
with the \textit{famm.ad} performing particularly well when the additive model is correctly specified.
The \textit{FOS-DNN} exhibits relatively high prediction errors similar to those of \textit{famm.nl}.
In most settings with Models 2 and 3, 
\textit{FUA} and \textit{famm.nl} achieve the lowest prediction errors. 
The \textit{FOS-DNN} performs slightly worse than these methods.

This suggests that in datasets with low spatial inhomogeneity and low anisotropy, the high adaptivity of the proposed method may lead to slightly lower performance compared to basis function-based methods when handling functions with uniform complexity, particularly in small-sample scenario.
Due to model misspecification,
the \textit{famm.ad} provides significantly higher prediction errors than the other nonlinear methods. 
The linear method \textit{fosr} consistently produces high prediction errors across all settings.

Among the three distributions of $(X_1, X_2, X_3)$, 
the MISPEs are the smallest for Xtype-1 in most nonlinear methods.
As noted by \cite{luo2023nonlinear}, one possible explanation is that for Xtype-1, 
$(X_1,X_2,X_3)$ is uniformly distributed over the bounded region $[0,1]^3$.
This uniform distribution ensures that data points are evenly spread across the entire region, 
making it easier to estimate the overall structure of the true function.
In contrast, 
the other two settings assume that $(X_1, X_2, X_3)$ follows a multivariate normal distribution. 
In these settings, the number of data points decreases rapidly as the distance from the origin $(0,0,0)$ increases.
Thus, the structure near the origin is well-estimated, 
but estimation becomes challenging in regions far from the origin.
The \textit{FUA}, \textit{famm.ad} and \textit{famm.nl}, which rely on basis functions, are capable of capturing the overall structure of the true function even when the data is not densely concentrated, demonstrating particularly good performance for Xtype-1.
On the other hand, while the \textit{FOS-DNN} struggles to capture the overall structure with a limited number of samples, it efficiently concentrates the model's expressive power where it is needed. 
As demonstrated by the results for Xtype-2 and Xtype-3 in Model 3, it could outperform other methods in settings where the estimation of local features is crucial.

\begin{table}[h]
    \centering
    \caption{Averages (and standard deviations) of the MISPEs from 100 replicates for all settings and methods in Scenario 1 ($N_{\text{train}} = 200$)}
     \label{tab:1}
     \resizebox{\textwidth}{!}{
    \begin{tabular}{|c|c|c|c|c|c|c|}
  \hline
   \textbf{X-type} &  \textbf{Model} &  \textbf{FOS-DNN} & \textbf{FUA} &  \textbf{famm.nl} &  \textbf{famm.ad} &  \textbf{fosr} \\
  \hline \hline 
   1 &  1 &  0.014 (0.001) &  \textbf{0.010 (0.000)} &  0.013 (0.001) &  \textbf{0.010 (0.000)} &  0.126 (0.005) \\
   &  2 &  0.027 (0.003) &  \textbf{0.013 (0.001)} &  0.021 (0.003) &  0.677 (0.035) &  1.033 (0.019)\\
   &  3 &  0.015 (0.001) &  \textbf{0.013 (0.003)} &  0.016 (0.003) &  0.900 (0.050) &  1.039 (0.005)  \\ \hline
   2 &  1 &  0.026 (0.007) & \textbf{0.014 (0.002)} &  0.032 (0.005) & \textbf{0.014 (0.008)} &  0.466 (0.018)\\
   &  2 &  0.062 (0.015) &  0.037 (0.015) &  \textbf{0.032 (0.007)} &  0.897 (0.095) &  1.037 (0.023)\\
   &  3 &  \textbf{0.021 (0.003)} &  \textbf{0.021 (0.005)} &  0.026 (0.005) &   0.935 (0.076) &  1.037 (0.027)\\
  \hline
   3 &  1 &  0.023 (0.006) &  0.014 (0.002) &  0.027 (0.006) &  \textbf{0.013 (0.004)} &  0.387 (0.022) \\
  &  2 &  0.055 (0.013) &  0.038 (0.071) &  \textbf{0.034 (0.007)} &  0.545 (0.052)  &  1.040 (0.019) \\
  &  3 &  \textbf{0.023 (0.003)} &  0.032 (0.018) &  0.031 (0.008) &  0.904 (0.074)  &  1.032 (0.026) \\ \hline 
\end{tabular}}
\end{table}

\subsection{Scenario 2} \label{subsec : Simulation 3}
In Scenario 2, we consider more complex models characterized by high spatial inhomogeneity and anisotropy. We generate the predictors $(X_1, \ldots, X_5)$ following the same procedure described as before. We consider the following three models for the true function $f^\circ$: 
\begin{align*}
    &\textbf{Model 1}  : f^\circ(X_1, X_2, X_3, X_4 ,X_5, t) = f_{\text{high,1}}(X_1, t) + f_{\text{low,A}}(X_1,\ldots, X_5, t), \\
    &\textbf{Model 2} : f^\circ(X_1, X_2, X_3, X_4 ,X_5, t) = 
    f_{\text{high,1}}(X_1, t) + f_{\text{high,2}}(X_2, t) + f_{\text{low,B}}(X_1,\ldots, X_5, t),\\
    &\textbf{Model 3} : f^\circ(X_1, X_2, X_3, X_4 ,X_5, t) = 
    f_{\text{high,1}}(X_1, t) + f_{\text{high,2}}(X_2, t) + f_{\text{high,3}}(X_3, t) \\
    &\qquad \qquad \qquad \qquad \qquad \qquad \qquad \qquad \qquad + f_{\text{low,C}}(X_1,\ldots, X_5, t),
\end{align*}
where $f_{\text{high}, i}$ represents a function exhibiting high spatial inhomogeneity in $X_i$, whereas $f_{\text{low,} \cdot}$ represents a function with low spatial inhomogeneity across all directions (see Appendix \ref{sec: Appendix C} for the specific functional structures). 
In Model 1, high spatial inhomogeneity is observed only in 
$X_1$, making it the model with the highest anisotropic smoothness among the three models. 
In contrast, Models 2 and 3 include two and three axes, respectively, with high spatial inhomogeneity, reducing overall anisotropic smoothness (see Figure \ref{fig:Simulation function 3} of Appendix \ref{sec: Appendix C}).  
The predictors are generated based on their joint distributions, and $Y(t)$ is then generated following the same procedure described as before. Sample response curves for each setting are provided in Figure \ref{fig:Simulation 3} of Appendix \ref{sec: Appendix C}.

Table \ref{tab:5} presents the averages and standard deviations of the MISPEs over 50 repetitions for all settings when $N_{\text{train}} = 5000$ (for the results when $N_{\text{train}} = 1000$, see Table \ref{tab:12} of Appendix \ref{sec: Appendix C}).
For the \textit{FOS-DNN}, we select $W = 32$, $L = 6$, and $\alpha = 10^{-5}$. Additionally,
for the \textit{FUA}, \textit{famm.nl}, and \textit{famm.ad}, we increase the number of basis functions by 10 compared to the default settings to address the high spatial inhomogeneity of the data. 
The \textit{FOS-DNN} achieves the smallest prediction errors across all settings, followed by the \textit{FUA}.
For Models 1 and 2, which exhibit higher anisotropic smoothness, the \textit{FOS-DNN} achieves MISPEs that are 2.4 to 4.1 times smaller than those of \textit{FUA}. 
In contrast, for Model 3, which has low anisotropic smoothness, the MISPEs of the \textit{FOS-DNN} are 1.3 to 2.8 times smaller than those of \textit{FUA}, resulting in a smaller performance gap. 
The \textit{FOS-DNN} maintains low prediction errors with the same parameter settings for width $W$ and depth $L$ as in Scenario 1, which features low spatial inhomogeneity and low anisotropic smoothness. In contrast, other methods, even with a sufficient increase in the number of basis functions, demonstrate high prediction errors compared to the \textit{FOS-DNN}. 
These results demonstrate the proposed method's superior adaptive capabilities, emphasizing its enhanced effectiveness in settings with high spatial inhomogeneity and anisotropic smoothness.

\begin{table}[h]
    \centering
    \caption{Averages (and standard deviations) of the MISPEs from 50 replicates for all settings and methods in Scenario 2 ($N_{\text{train}} = 5000$)}
    \label{tab:5}
\begin{tabular}{|c|c|c|c|c|c|c|}
  \hline
   \textbf{X-type}  & \textbf{Model} & \textbf{FOS-DNN} &  \textbf{FUA} &  \textbf{famm.ad}  &  \textbf{fosr} \\
  \hline \hline 
   1  & 1 & \textbf{0.011 (0.000)} &  0.026 (0.004) &  0.418 (0.023) &  1.014 (0.032)  \\
     & 2 & \textbf{0.013 (0.001)} &  0.053 (0.007) &  0.675 (0.031) &  1.014 (0.041)  \\
     & 3 & \textbf{0.069 (0.071)} & 0.194 (0.022) &  0.279 (0.013) &  1.004 (0.040)  \\ \hline
   2  & 1 & \textbf{0.011 (0.000)} & 0.036 (0.003) &  0.261 (0.013) &  1.015 (0.037) \\
     & 2 & \textbf{0.039 (0.087)} &  0.147 (0.023) &  0.642 (0.037) &  1.020 (0.055) \\
     & 3 & \textbf{0.071 (0.008)} &  0.097 (0.011) &  0.120 (0.010) &  1.035 (0.075) \\ \hline
   3    & 1 & \textbf{0.011 (0.000)} & 0.034 (0.002) &  0.269 (0.014) &  1.008 (0.033)  \\ 
       & 2 & \textbf{0.043 (0.072)} &  0.169 (0.042) &  0.678 (0.038) &  0.997 (0.059)  \\ 
       & 3 & \textbf{0.083 (0.008)} & 0.110 (0.007) &  0.180 (0.013) &  1.015 (0.064)  \\ \hline 
\end{tabular} 
\end{table}

\subsection{Scenario 3} \label{subsec : Simulation 4}
In Scenario 3, we examine models characterized by high spatial inhomogeneity and high anisotropic smoothness, which have higher dimensions compared to Scenario 2.
We set $d=10$ and generate the predictors $(X_1, \ldots, X_{10})$ using the same procedure. We consider the following three models for the true function $f^\circ$:
\begin{align*}
    &\textbf{Model 1} : f^\circ(X_1,\ldots, X_{10}, t) =  g_{\text{high,1}}(X_1, t)  + g_{\text{low,A}}(X_1,\ldots, X_{10},t) \\
    &\textbf{Model 2} : f^\circ(X_1,\ldots, X_{10}, t) = g_{\text{high,1}}(X_1, t) + g_{\text{high,2}}(X_2, t) + g_{\text{high,3}}(X_3, t) \\
    &\qquad \qquad \qquad \qquad \qquad \qquad \qquad + g_{\text{low,B}}(X_1,\ldots, X_{10},t)\\
    &\textbf{Model 3} : f^\circ(X_1,\ldots, X_{10}, t) = g_{\text{high,1}}(X_1, t) + g_{\text{high,2}}(X_2, t) + g_{\text{high,3}}(X_3, t)  \\ 
    &\qquad \qquad \qquad \qquad  \qquad \qquad \qquad + g_{\text{high,4}}(X_4, t) +  g_{\text{high,5}}(X_5, t) + g_{\text{low,C}}(X_1,\ldots, X_{10},t)
\end{align*}
where $g_{\text{high}, i}$ represents a function exhibiting high spatial inhomogeneity in $X_i$, whereas $g_{\text{low},\cdot}$ represents a function with low spatial inhomogeneity across all directions (see Appendix \ref{sec: Appendix C} for the specific functional structures). 
In Model 1, high spatial inhomogeneity is observed only in 
$X_1$, making it the model with the highest anisotropic smoothness among the three models. 
Models 2 and 3 incorporate three and five axes, respectively, exhibiting high spatial inhomogeneity (see Figure \ref{fig:Simulation function 4} of the Appendix \ref{sec: Appendix C}). 
The predictors are generated based on their joint distributions, and $Y(t)$ is then generated following the same procedure described as before. Sample response curves for each setting are provided in Figure \ref{fig:Simulation 4} of the Appendix \ref{sec: Appendix C}.

For each setting, we conducted 20 repetitions with a training data size of $N_{\text{train}} = 10000$.
Table \ref{tab:6} presents the means and standard deviations of the MISPEs over 20 repetitions for all settings when $N_{\text{train}}=10000$ (for the results when $N_{\text{train}} = 5000$, see Table \ref{tab:15} of Appendix \ref{sec: Appendix C}). 
For the \textit{FOS-DNN}, we set $W=32$, $L=7$, and $\alpha=10^{-3}$. 
Additionally, for \textit{FUA}, \textit{famm.nl}, and \textit{famm.ad}, the number of basis functions is also increased by 10 compared to the default settings to address the high spatial inhomogeneity of the data. 
The \textit{FOS-DNN} achieves the best predictive performance across all settings, followed by the \textit{FUA}.
For Models 1 and 2, which exhibit high anisotropy, the MISPEs of the \textit{FOS-DNN} are 5.2 to 11.4 times smaller than those of \textit{FUA}, representing a significantly more significant improvement compared to the results in Scenario 3. 
In contrast, for Model 3, which has low anisotropy, the MISPEs of the \textit{FOS-DNN} are only 1.3 to 1.8 times smaller than those of the \textit{FUA}, indicating a small performance gap. 
Notably, the \textit{FOS-DNN} achieves these results with approximately 61\% fewer total parameters than the \textit{FUA} (\textit{FOS-DNN}: 4641, \textit{FUA}: 7500). Other methods exhibit lower predictive accuracy compared to the \textit{FOS-DNN}.
These results suggest that the adaptive capabilities of the proposed method become more beneficial in scenarios with higher spatial inhomogeneity and anisotropic smoothness, particularly in higher-dimensional settings.

\begin{table}[h]
    \centering
    \caption{Averages (and standard deviations) of the MISPEs from 20 replicates for all settings and methods in Scenario 3 ($N_{\text{train}} = 10000$)}
    \label{tab:6}
\begin{tabular}{|c|c|c|c|c|c|c|}
  \hline
   \textbf{X-type}  & \textbf{Model} & \textbf{FOS-DNN} &  \textbf{FUA} &  \textbf{famm.ad}  & \textbf{fosr} \\
  \hline \hline 
   1  & 1 & \textbf{0.017 (0.000)}  & 0.060 (0.004)  & 0.290 (0.016)  &  1.024 (0.035)  \\
     & 2 &  \textbf{0.019 (0.002)} & 0.183 (0.024)  & 0.780 (0.045)  &  1.030 (0.063)   \\
     & 3 & \textbf{0.253 (0.096)} & 0.305 (0.021) & 0.278 (0.011) & 0.998 (0.031)      \\ \hline
   2  & 1 & \textbf{0.016 (0.000)}  &  0.114 (0.013) & 0.157 (0.007) & 1.014 (0.030)   \\
     & 2 & \textbf{0.030 (0.020)}  &  0.178 (0.015) & 0.804 (0.073) & 1.064 (0.100)  \\
     & 3 & \textbf{0.051 (0.003)} & 0.095 (0.006)  & 0.104 (0.011)  & 0.991 (0.051)  \\ \hline
   3    & 1 & \textbf{0.015 (0.000)} & 0.082 (0.016)  & 0.128 (0.008)  & 1.010 (0.016)  \\ 
       & 2 & \textbf{0.023 (0.002)}  & 0.263 (0.026)  & 0.788 (0.076)  & 1.047 (0.101) \\ 
       & 3 & \textbf{0.059 (0.004)} & 0.079 (0.004) &  0.168 (0.014)  & 1.043 (0.064)  \\ \hline 
\end{tabular} 
\end{table}

\section{Application: Ground reaction force data} \label{sec : Application}
\leavevmode

In this section, we apply the proposed method to analyze the ground reaction force (GRF) dataset (\citealp{ZHANG202456}). 
The GRF signals capture how the reaction force exerted by the ground on the human body changes over time and can be analyzed as functional data.
In the field, such analyses are essential to assess the risk of career-threatening injuries and to understand their relationship with individual motor dynamics.
This study aims to evaluate how physical information, including the presence or absence of anterior cruciate ligament (ACL) injuries, can explain the GRF data.

\cite{terada2020} examined similar datasets in the context of classification problems, identifying athletes who have experienced ACL injuries and those at high risk of such injuries. Notably, some athletes identified as high-risk during these studies subsequently sustained ACL injuries or severe ankle sprains after the experiments.
These findings suggest a potential causal relationship between GRF signals and injuries such as ACL. Investigating the structure and nature of this relationship could provide valuable insights for both academic and practical applications. 

The proposed method is particularly well-suited for analyzing this dataset, as it can adaptively capture local variations and incorporate covariate information, essential for modeling the complex interactions between physical characteristics and GRF signals. For example, local variations in GRF signals, such as abrupt changes during the landing phase, are considered critical biomechanical indicators that may be associated with injury risks like ACL tears. As shown in Figure \ref{fig:GRF}, the GRF signals exhibit apparent local variations, highlighting the applicability of the proposed method to this dataset.

This dataset was collected from 222 healthy subjects (age: 14.5 $\pm$ 2.2 years; 96 males and 126 females) with no history of lower limb injuries within six months before the experiment. 
The procedure for this experiment was approved by the local ethics board, and prior to data collection, each subject provided written informed consent.
In addition to GRF data for each participant, the dataset includes 25 covariates, such as age and height, and information on leg injuries, including the presence or absence of ACL injury history. A complete list of predictors with their descriptions and statistical summaries is provided in Table \ref{tab:variables_summary} of Appendix \ref{sec: Appendix C}.

The experimental task involved a single-legged drop landing. Participants were instructed to jump forward from a wooden platform 0.2 meters high and land on a force plate using either leg. Upon landing, they were required to maintain a quiet, single-legged stance for at least 8 seconds. During each landing, participants were asked to cross their arms over their chest and keep them in that position throughout the trial to minimize the impact of arm movements (see \cite{ZHANG202456} for details on GRF data measurement). Trials in which participants lost balance, stepped off the force plate, or failed to maintain the single-legged stance were deemed unsuccessful and excluded from further analysis. For each participant, we collected data from successful trials six times for both the left and right legs. Data points identified as clear measurement errors were removed, resulting in $N = 2465$ valid observations.

To evaluate the performance of the proposed method, we use the GRF data as the response variable and the normalized 25 covariates as the predictor variables. For the GRF data, we consider the vertical component during the first approximately one second after the initial impact. Additionally, each subject's body weight was normalized to eliminate its effect.

We repeat the following procedure 10 times. In each iteration, we randomly split the 2465 observations into a training set of 1972 observations and a test set of 493 observations. The estimated model is then applied to the test set to calculate the MISPE. We compare our method with the \textit{FUA} and \textit{fosr}. The \textit{famm.ad} method was not applied due to many categorical covariates and limited unique combinations.

Table \ref{tab:7} presents the averages and standard deviations of the MISPEs over 10 repetitions. For the proposed method, we set  $W = 32$, $L = 7$, and $\alpha = 10^{-5}$. For the \textit{FUA} and the \textit{fosr}, the number of basis functions is increased by 10 from the default setting to capture local variations (See Table \ref{tab:19} in Appendix \ref{sec: Appendix C} for \textit{FUA} with different numbers of basis functions).  The results indicate that the proposed method and the \textit{FUA} demonstrate comparable predictive accuracy. However, in terms of the total number of parameters, the proposed method requires 5,121 parameters, while the \textit{FUA} requires 16,250 parameters, approximately three times more. This suggests that the proposed method achieves higher performance and greater efficiency regarding the number of parameters.

\begin{figure}
    \centering
    \includegraphics[width=0.5\linewidth]{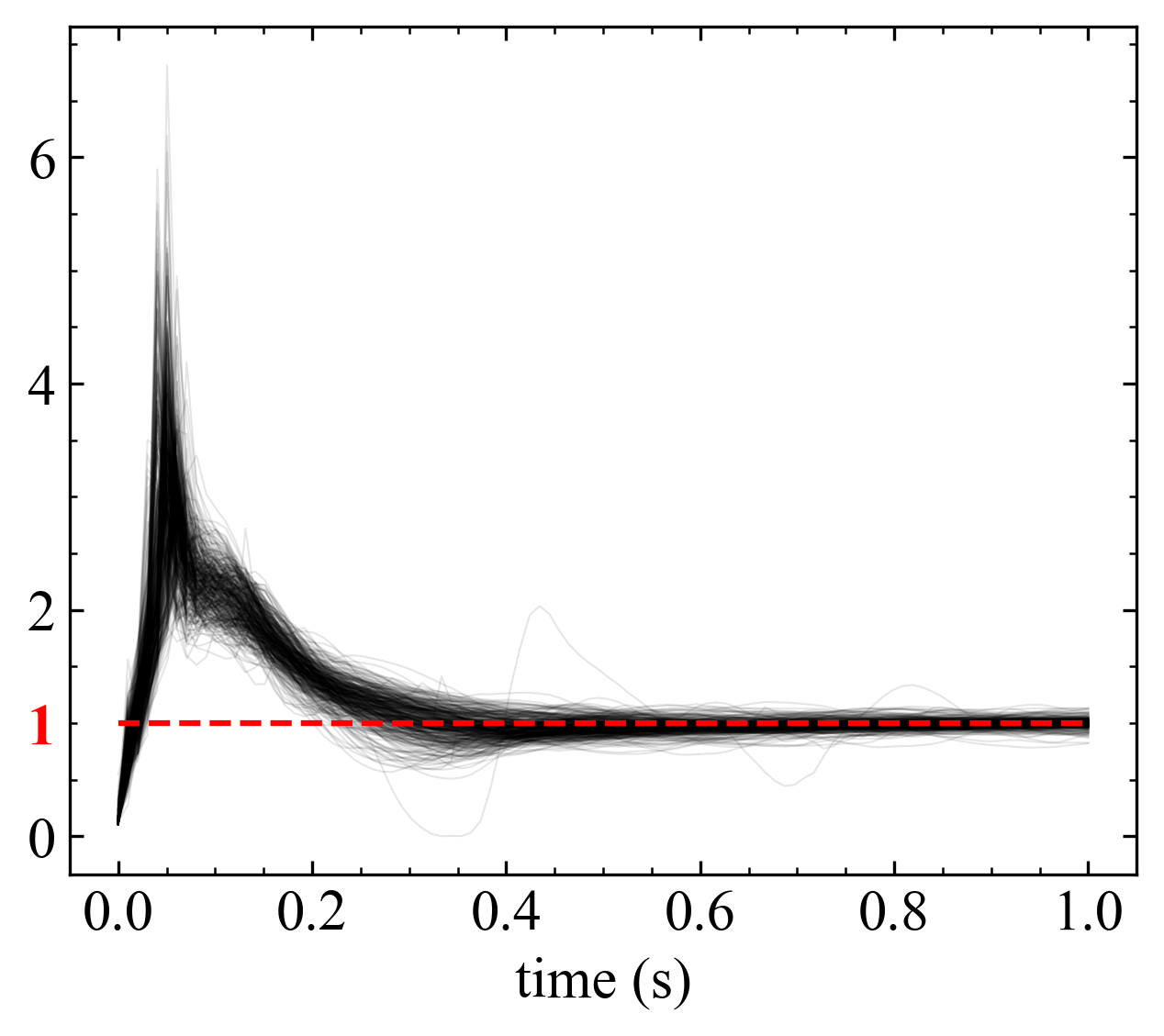}
    \caption{Sample curves of the ground reaction force (GRF) data. The black line represents the data, and the red dashed line represents a horizontal line at 1.}
    \label{fig:GRF}
\end{figure}

\begin{table}[h]
    \centering
    \caption{Averages (and standard deviations) of the MISPEs from 10 replicates, in the GRF data.}
    \label{tab:7}
\begin{tabular}{|c|c|c|c|c|c|c|}
  \hline
   \textbf{FOS-DNN} & \textbf{FUA} &  \textbf{fosr} \\
  \hline \hline 
   \textbf{0.027 (0.001)} &  \textbf{0.027 (0.001)} &  0.047 (0.001) \\ \hline 
\end{tabular} 
\end{table}

\section{Discussion}
\leavevmode

In this study, we propose a novel method for nonlinear Function-on-Scalar (FOS) regression utilizing deep learning. Unlike traditional approaches that rely on basis functions at a fixed spatial scale, the proposed method adaptively learns optimal representational structures directly from the data. This adaptivity enables efficient use of parameters and improves estimation accuracy, particularly for datasets with localized variations or high anisotropy. These features are demonstrated through theoretical analyses and numerical simulations, where the proposed method consistently outperforms previous methods in scenarios with high spatial inhomogeneity and anisotropic smoothness. The adaptive learning capability of the proposed approach is crucial for capturing the complexity of modern, highly structured datasets.

While this study focuses on enhancing predictive performance, practical applications often require interpretability to understand how predictor variables influence the functional response. 
This highlights the importance of developing semiparametric FOS regression models that integrate deep neural networks capable that are adaptively estimating nonlinear components. 
Furthermore, from a theoretical perspective, analyzing the case where functional data are observed discretely in more practical settings and demonstrating that the proposed method achieves minimax optimality would significantly strengthen its theoretical foundation and promote its broader application across various fields.

\section*{Acknowledgments}
\leavevmode

The authors express their sincere gratitude to Dr. Issei Ogasawara for generously providing the GRF data.
The authors also gratefully acknowledge JSPS KAKENHI (Grant Numbers JP23K24737 and JP24K14855), and the MEXT Project for Seismology Toward Research Innovation with Data of Earthquakes (STAR-E, JPJ010217).

\newpage
\begin{appendix}

\section{Proofs}\label{sec : Appendix A} 
\subsection{Proof of \Cref{thm1}}
\leavevmode

In this section, we prove \Cref{thm1}. To begin, we first introduce the definition of the covering number.

\begin{definition} \label{def:covering number}
    Let $\mathcal{C}$ be an arbitrary function space, and let $\hat{d}: \mathcal{C} \times \mathcal{C} \rightarrow \mathbb{R}+$ be a distance function. The covering number $\mathcal{N}(\epsilon, \mathcal{C}, \hat{d})$ is defined as the smallest integer $N$ for which there exists a finite set $\{g_j\}_{j=1}^N \subset \mathcal{C}$ such that
\begin{equation*} 
\sup_{g \in \mathcal{C}} \min_{j \in \{1,\ldots, N\}} \hat{d}(g, g_j) \leq \epsilon. 
\end{equation*}

We call $\mathcal{N}(\epsilon, \mathcal{C}, \hat{d})$ the $\epsilon$-covering number of $\mathcal{C}$.
\end{definition}

The key lemma for the proof of \Cref{thm1} is the following.
\begin{lemma} 
\label{lemma : main}
    For any estimator $\hat{f}_n \in \Phi(L,W,S,B) := \Phi$, we define
    \begin{align*}
    &\Delta_n(\hat{f}_n) \notag \\
    &:= E_{D_n}\left\lbrack \frac{1}{n} \sum_{i=1}^n \int_0^1 (Y_i(t) - \hat{f}_n(\boldsymbol{X}_i, t))^2 dt - \inf_{f \in \Phi}  \frac{1}{n} \sum_{i=1}^n \int_0^1 (Y_i(t) - f(\boldsymbol{X}_i, t))^2 dt\right\rbrack  \\
    &\qquad \text{and} \ R(\hat{f}_n, f^{\circ}) := E_{D_n}\left\lbrack \int_0^1 \|f^{\circ}(\boldsymbol{X}, t) - \hat{f}_n(\boldsymbol{X},t)\|_{L^2(p_{\boldsymbol{X}})}^2 dt \right\rbrack. 
    \end{align*}
    For the sake of simplicity, we denote $\Delta_n(\hat{f}_n)$ briefly by $\Delta_n$.
    Additionally, suppose that for some $F \geq 1$, it holds that ${f^{\circ}} \cup \Phi \subset \{f : \Omega \rightarrow [-F, F]\}$. Moreover, assume there exists a constant $R > 0$ such that the probability density function $p_{\boldsymbol{X}}$ of $\boldsymbol{X}$ satisfies $\|p_{\boldsymbol{X}} \|_\infty \leq R$.
    Then, under the condition that $\mathcal{N}_n := \mathcal{N}(\delta, \Phi, |\cdot|\infty) \geq 3$, for all $\epsilon, \delta \in (0,1]$, the following inequality holds: 
     \begin{align}
    &R(\hat{f}_n, f^{\circ}) \notag \\
    &\leq  (1+\epsilon)^2 \left(R \inf_{f \in \Phi} \|f - f^\circ\|_{L^2(\Omega)}^2  \right. \notag \\
    &\qquad \qquad \qquad + F^4 \frac{(4\sigma^2 +12) \log \mathcal{N}_n + (8\sigma^2 + 70)}{n \epsilon}  + 2 \delta K + 26 \delta F^2 + 6 \delta \sigma \biggr), \label{eq:lemma1 result}
    \end{align}
    where $K := E\left\lbrack \int_0^1 |\xi_1(t)| dt\right\rbrack < \infty$.
\end{lemma}

The proof of \Cref{lemma : main} can be found in Appendix \ref{sec : Appendix B}. In this paper, we assume that $\Delta_n = 0$ always holds. 
In the upper bound of $R(\hat{f}_n, f^{\circ})$ in (\ref{eq:lemma1 result}), the quantities to be evaluated are: (a) $\inf_{f \in \Phi} \|f - f^{\circ}\|_{L^2(\Omega)}^2$ and (b) $\log \mathcal{N}_n$.

First, for (a), the following result is established in Proposition 2 of \citet{suzuki2021deep}.

\begin{lemma}{\textnormal{(Proposition 2 of \cite{suzuki2021deep}}}) \label{lemma : suzuki prop2} \par 
    Assume the same condition as in \Cref{thm1}.  we can bound the approximation error as
    \begin{equation*}
        \sup_{f^\ast \in B_{p,q}^{\beta}(\Omega)} \inf_{f \in \Phi(L_1,W_1,S_1,B_1)} \|f^\ast - f\|_{L^2(\Omega)} \lesssim N^{- \tilde{\beta}}. \label{eq:approximation theorem}
    \end{equation*}
\end{lemma}

Additionally, for (b), its upper bound is also provided in Appendix C of \cite{suzuki2021deep}.
\begin{lemma}{\textnormal{(Appendix C of \cite{suzuki2021deep}}}) \label{lemma : suzuki apendix c}\par 
    Let $\Phi$ be as defined in \Cref{lemma : main}. Then, we obtain
    \begin{equation*}
        \log \mathcal{N}_n(\delta, \Phi, \|\cdot\|_{\infty}) \lesssim \mathcal{N}_n \log \mathcal{N}_n [(\log \mathcal{N})^2 + \log(\delta^{-1})]. \label{eq:covering number}
    \end{equation*}
\end{lemma}

\textit{Proof of Thorem 1}. By substituting the results of \Cref{lemma : suzuki prop2} and \Cref{lemma : suzuki apendix c} into (\ref{eq:lemma1 result}) in \Cref{lemma : main} and setting $\delta = 1/n$, we have
\begin{align}
    &E_{D_n} \left\lbrack\int_0^1 (f^o(\boldsymbol{X}, t) - \hat{f}_n(\boldsymbol{X}, t))^2 dt \right \rbrack \notag \\
    &\lesssim    R N^{-2 \tilde{\beta}} + \frac{N \log N \log n}{n \epsilon} + \frac{N (\log N)^3}{n \epsilon} +  \frac{1}{n \epsilon} + \frac{1}{n}. \label{eq:estimation upper bound}
\end{align}

Here, the right hand side is minimized by setting $N \approx n^{\frac{1}{2 \tilde{\beta} + 1}}$, the convergence rate is bounded by
\begin{equation*}
    n^{- \frac{2 \tilde{\beta}}{2 \tilde{\beta} + 1}} (\log n)^3. \label{eq:estimation error rate}
\end{equation*}
This completes the proof.

\subsection{Proof of \Cref{lemma : main}} \label{sec : Appendix B}
\leavevmode

We prove \Cref{lemma : main} in a manner similar to \cite{schmidt2020Nonpara}.
For simplicity, we denote $E_{D_n}[\cdot] = E[\cdot]$ and $\int_0^1 = \int$. 
We define $\|g\|_n^2 = \frac{1}{n}\sum_{i=1}^n \int g(\boldsymbol{X}_i, t)^2 dt$.
For any estimator $\tilde{f}$, we let $\hat{R}_n(\tilde{f}, f^{\circ}) = E[\|\tilde{f} - f^{\circ}\|_n^2]$. 
First, in the case where $\log \mathcal{N}_n \geq n$, it is straightforward to see that (\ref{eq:lemma1 result}) holds by using $R(\hat{f}_n, f^{\circ}) \leq 4F^2$. 
Therefore, we proceed under the assumption that $\log \mathcal{N}_n \leq n$.
The proof is carried out in the following three steps: (I) - (III).

\begin{itemize}
    \item[(I)] We evaluate $R(\hat{f}_n, f^{\circ})$ using the empirical risk $\hat{R}_n(\hat{f}_n, f^\circ)$: 
    \begin{align*}
        &R(\hat{f}_n, f^{\circ}) \leq  (1 + \epsilon)\left(\hat{R}_n(\hat{f}, f^{\circ})+ (1+\epsilon) F^4 \frac{12 \log \mathcal{N}_n+ 70}{n \epsilon} + 26\delta F^2 \right).
    \end{align*}
    \item[(II)] For any estimator $\tilde{f} \in \Phi$, we obtain the inequality
    \begin{equation*}
        \left|E \left\lbrack \frac{2}{n} \sum_{i=1}^n \int \xi_i(t) \tilde{f}(\boldsymbol{X}_i, t)\right\rbrack \right| \leq  2 \sqrt{\frac{\hat{R}_n(\tilde{f}, f^{\circ}) (2\sigma^2 \log \mathcal{N}_n + 4\sigma^2)}{n}}  + 2 \delta K + 6\delta \sigma.
    \end{equation*}
    \item[(III)]  The following inequality holds as an upper bound for $\hat{R}_n(\hat{f}, f^{\circ})$: 
    \begin{align*}
        &\hat{R}_n(\hat{f}_n, f^{\circ}) \\
        &\leq (1+\epsilon) \left( \inf_{f \in \Phi} E \left\lbrack \int (f(\boldsymbol{X},t) - f^{\circ}(\boldsymbol{X},t))^2 dt\right\rbrack + \right. \\
        &\hspace{5cm} \left. 2\delta K + 6\delta \sigma + F^2 \frac{4\sigma^2 \log \mathcal{N}_n + 8\sigma^2}{n \epsilon} + \Delta_n \right) \\
        &\leq (1+\epsilon) \left(R \inf_{f \in \Phi} \|f - f^{\circ}\|_{L^2(\Omega)}^2 +  2\delta K + 6\delta \sigma + F^2 \frac{4\sigma^2 \log \mathcal{N}_n + 8\sigma^2}{n\epsilon} + \Delta_n  \right).
    \end{align*}
\end{itemize}
The upper bound in (\ref{eq:lemma1 result}) can be obtained by combining (I) and (III).

\textit{Proof of Lemma 1}.
We  prove (I) through (III) step by step below.

\textit{(I)}: Take a $\delta$-covering of $\Phi$. Let the covering center of each ball be denoted by $f_j$. By the properties of covering, for the estimator $\hat{f}_n$, there exists an index $j^{\ast}$ such that $|\hat{f}_n - f_{j^{\ast}}|_{\infty} \leq \delta$. Let $\{\boldsymbol{X}_i^{\prime} \}_{i=1}^n$ be a sequence of random variables generated independently from the same distribution as $\boldsymbol{X}$ and $\{\boldsymbol{X}_i\}_{i=1}^n$. Using that $|f_j|_{\infty}, |f^{\circ}|_{\infty}, \delta \leq F$, we have

\begin{align}
    &|R(\hat{f}_n, f^{\circ}) - \hat{R}_n(\hat{f}_n, f^{\circ})| \notag\\
    &\leq E  \left\lbrack \left| \frac{1}{n} \sum_{i=1}^n \int \left\lbrace (f_{j^\ast}(\boldsymbol{X}_i^\prime, t) -f^{\circ}(\boldsymbol{X}_i^\prime, t))^2 - (f_{j^\ast}(\boldsymbol{X}_i, t) - f^{\circ}(\boldsymbol{X}_i, t))^2\right\rbrace dt \right| \right\rbrack  \notag\\
    &\quad + E\left\lbrack  \frac{2}{n} \sum_{i=1}^n \int \left\lbrace \underset{\leq \delta}{\underbrace{|\hat{f}_n(\boldsymbol{X}_i^\prime, t) - f_{j^\ast}(\boldsymbol{X}_i^\prime, t)|}} \cdot \underset{\leq 2F}{\underbrace{|(f_{j^\ast}(\boldsymbol{X}_i^\prime, t) - f^{\circ}(\boldsymbol{X}_i^\prime, t)|}}  \right.\right. \notag \\
    & \quad \quad \quad \quad \quad \left.\left. +  \underset{\leq \delta}{\underbrace{|\hat{f}_n(\boldsymbol{X}_i, t) - f_{j^\ast}(\boldsymbol{X}_i, t)|}} \cdot \underset{\leq 2F}{\underbrace{|f_{j^\ast}(\boldsymbol{X}_i, t) - f^{\circ}(\boldsymbol{X}_i, t)|}} \right\rbrace dt \right\rbrack  \notag\\
    &\quad \quad +  E \left\lbrack \frac{1}{n} \sum_{i=1}^n \int \underset{\leq \delta^2}{\underbrace{\left| (\hat{f}_n(\boldsymbol{X}_i^\prime, t) - f_{j^\ast}(\boldsymbol{X}_i^\prime, t))^2 - (\hat{f}_n(\boldsymbol{X}_i, t) - f_{j^\ast}(\boldsymbol{X}_i, t))^2  \right|}} dt \right\rbrack \notag \\
    &\leq E\left\lbrack \left| \frac{1}{n} \sum_{i=1}^n g_{j^{\ast}}(\boldsymbol{X}_i, \boldsymbol{X}_i^{\prime})\right|\right\rbrack + 9\delta F, 
\end{align}
where $g_{j}(\boldsymbol{X}_i, \boldsymbol{X}_i^{\prime}) := \int \{f_j(\boldsymbol{X}_i^{\prime}, t) - f^{\circ}(\boldsymbol{X}_i^{\prime}, t)^2 - (f_j(\boldsymbol{X}_i, t) - f^{\circ}(\boldsymbol{X}_i, t))^2\} dt$.
Define $r_j :=  \sqrt{n^{-1} \log \mathcal{N}_n} \lor (E[\int (f_j(\boldsymbol{X},t) - f^{\circ}(\boldsymbol{X},t))^2 dt])^{1/2}$. 
Using that $(E[|a+b|^2])^{1/2} \leq (E[a^2])^{1/2} + (E[b^2])^{1/2}$, we obtain
\begin{align*}
    r_{j^{\ast}} &= \sqrt{n^{-1} \log \mathcal{N}_n} \lor \left(E\left\lbrack \int (f_{j^{\ast}}(\boldsymbol{X},t) - f^{\circ}(\boldsymbol{X},t))^2 dt | \{\boldsymbol{X}_i, Y_i \}_{i=1}^n \right\rbrack \right)^{1/2}\\
    &= \sqrt{n^{-1} \log \mathcal{N}_n} \\
    & \quad + \left(E\left\lbrack \int (f_{j^{\ast}}(\boldsymbol{X},t) -\hat{f}_n(\boldsymbol{X},t) + \hat{f}_n(\boldsymbol{X},t) - f^{\circ}(\boldsymbol{X},t))^2 dt | \{\boldsymbol{X}_i, Y_i \}_{i=1}^n \right\rbrack \right)^{1/2}\\
    &\leq \sqrt{n^{-1} \log \mathcal{N}_n} + \left(E\left\lbrack \int \underset{\leq \delta^2}{\underbrace{(f_{j^{\ast}}(\boldsymbol{X},t) - \hat{f}_n(\boldsymbol{X},t))^2}} dt | \{\boldsymbol{X}_i, Y_i \}_{i=1}^n \right\rbrack  \right)^{1/2} \\
    & \quad \quad \quad \quad + \left(E\left\lbrack \int (\hat{f}_n(\boldsymbol{X},t) - f^{\circ}(\boldsymbol{X},t))^2 dt | \{\boldsymbol{X}_i, Y_i \}_{i=1}^n \right\rbrack  \right)^{1/2} \\
    & \leq \sqrt{n^{-1} \log \mathcal{N}_n} +  \left(E\left\lbrack \int (\hat{f}_n(\boldsymbol{X},t) - f^{\circ}(\boldsymbol{X},t))^2 dt | \{\boldsymbol{X}_i, Y_i \}_{i=1}^n \right\rbrack  \right)^{1/2} + \delta.
\end{align*}
Define random variables $U, T$ as 
\begin{align*}
    &U := \left(E\left\lbrack \int (\hat{f}_n(\boldsymbol{X},t) - f^{\circ}(\boldsymbol{X},t))^2 dt | (\boldsymbol{X}_i, Y_i)_i\right\rbrack \right)^{1/2} \\
     &\text{and} \ T := \max_j \left|\sum_{i=1}^n g_j(\boldsymbol{X}_i, \boldsymbol{X}_i^{\prime}) /(r_j F^2) \right|.
\end{align*}
Using the fact that
\begin{align*}
    E[U^2] &= E\left\lbrack E \left\lbrack \int (\hat{f}_n(\boldsymbol{X},t) - f^{\circ}(\boldsymbol{X},t))^2 dt |\{\boldsymbol{X}_i, Y_i \}_{i=1}^n \right\rbrack \right\rbrack \\
    &= E\left\lbrack \int (\hat{f}_n(\boldsymbol{X},t) - f^{\circ}(\boldsymbol{X},t))^2 dt\right\rbrack = R(\hat{f}_n, f^{\circ})
\end{align*}
and that $E[UV] \leq (E[U^2])^{1/2} (E[V^2])^{1/2}$ hold, we obtain
\begin{align}
    &|R(\hat{f}_n, f^{\circ}) - \hat{R}_n(\hat{f}_n, f^{\circ})| \notag\\
    &\leq E \left\lbrack \left|\frac{1}{n} \sum_{i=1}^n g_{j^{\ast}}(\boldsymbol{X}_i, \boldsymbol{X}_i^{\prime}) \right| \right\rbrack + 9\delta F \notag\\
    &= \frac{F^2}{n} E \left\lbrack r_{j^{\ast}} \cdot \left|\sum_{i=1}^n g_{j^{\ast}}(\boldsymbol{X}_i, \boldsymbol{X}_i^{\prime}) / (r_{j^{\ast}} F^2) \right| \right\rbrack + 9\delta F \notag\\
    &\leq \frac{F^2}{n} E \left\lbrack r_{j^{\ast}} T \right\rbrack + 9\delta F \notag\\
    &\leq \frac{F^2}{n} E[UT] + \frac{F^2}{n} \left(\sqrt{n^{-1} \log \mathcal{N}_n} + \delta \right)E[T] + 9\delta F \notag\\
    &\leq \frac{F^2}{n} (E[U^2])^{1/2} (E[T^2])^{1/2} +  \frac{F^2}{n} \left(\sqrt{n^{-1} \log \mathcal{N}_n} + \delta \right)E[T] + 9\delta F \notag\\
    &= \frac{F^2}{n} R(\hat{f}_n, f^{\circ})^{1/2}(E[T^2])^{1/2} +  \frac{F^2}{n} \left(\sqrt{n^{-1} \log \mathcal{N}_n} + \delta \right)E[T] + 9\delta F. \label{eq:difference estimation error}
\end{align}
Next, we derive the upper bounds for $E[T]$ and $E[T^2]$.
We verify the following for $g_j(\boldsymbol{X}_i, \boldsymbol{X}_i^{\prime})$.
\begin{align*}
    &\cdot E[g_j(\boldsymbol{X}_i, \boldsymbol{X}_i^\prime)] \\
    &\quad = E\left\lbrack \int (f_{j^\ast} (\boldsymbol{X}_i^\prime, t) - f^{\circ}(\boldsymbol{X}_i^\prime, t))^2  - (f_{j^\ast}(\boldsymbol{X}_i, t) - f^{\circ}(\boldsymbol{X}_i, t))^2 dt \right\rbrack \\
    &\quad= E \left\lbrack \int (f_{j^\ast} (\boldsymbol{X}_i^\prime, t) - f^{\circ}(\boldsymbol{X}_i^\prime, t))^2 dt \right\rbrack - E \left\lbrack \int (f_{j^\ast} (\boldsymbol{X}_i, t) - f^{\circ}(\boldsymbol{X}_i, t))^2 dt \right\rbrack = 0 \\
    &\cdot |g_j(\boldsymbol{X}_i, \boldsymbol{X}_i^\prime)| \\
    &\quad= \left|\int (f_{j^\ast} (\boldsymbol{X}_i^\prime, t) - f^{\circ}(\boldsymbol{X}_i^\prime, t))^2  - (f_{j^\ast}(\boldsymbol{X}_i, t) - f^{\circ}(\boldsymbol{X}_i, t))^2 dt  \right| \\
    &\quad\leq \int \underset{\leq (2F)^2}{\underbrace{|(f_{j^\ast} (\boldsymbol{X}_i^\prime, t) - f^{\circ}(\boldsymbol{X}_i^\prime, t))^2  - (f_{j^\ast}(\boldsymbol{X}_i, t) - f^{\circ}(\boldsymbol{X}_i, t))^2} }| dt \leq 4F^2 \\
    &\cdot \text{Var}(g_j(\boldsymbol{X}_i, \boldsymbol{X}_i^\prime)) \\
    &\quad= 2 \text{Var} \left(\int (f_j(\boldsymbol{X}_i^\prime, t) - f^{\circ}(\boldsymbol{X}_i^\prime, t))^2 dt \right) \\
    &\quad\leq 2 E\left\lbrack \left( \int (f_j(\boldsymbol{X}_i, t) - f^{\circ}(\boldsymbol{X}_i, t))^2 dt \right)^2 \right\rbrack \\
    &\quad= 2 E\left\lbrack \left(\int \underset{\leq 4F^2}{\underbrace{(f_j(\boldsymbol{X}_i, t) - f^{\circ}(\boldsymbol{X}_i, t))^2}} dt \right) \cdot  \left(\int (f_j(\boldsymbol{X}_i, t) - f^{\circ}(\boldsymbol{X}_i, t))^2 dt \right) \right\rbrack \\
    &\quad\leq 2 \cdot 4F^2  E\left\lbrack \int (f_j(\boldsymbol{X}_i, t) - f^{\circ}(\boldsymbol{X}_i, t))^2 dt \right\rbrack \leq 8F^2 r_j^2 
\end{align*}
Then, we observe that
\begin{align*}
    &E[g_j(\boldsymbol{X}_i, \boldsymbol{X}_i^{\prime})/(r_j F^2)] = 0, \quad |g_j(\boldsymbol{X}_i, \boldsymbol{X}_i^{\prime})/(r_j F^2)| \leq \frac{4F^2}{f_j F^2} = \frac{4}{r_j}, \\
    &\text{and} \ \text{Var}(g_j(\boldsymbol{X}_i, \boldsymbol{X}_i^{\prime})/(r_j F^2)) \leq \frac{8F^2r_j^2}{r_j^2 F^4} = \frac{8}{F^2}.
\end{align*}
We apply the following Bernstein's inequality to the random variables $\{g_j(\boldsymbol{X}_i, \boldsymbol{X}_i^{\prime})/(r_j F^2)\}_{i=1}^n$.
Bernstein's inequality states that for a sequence of independent random variables $U_1, \ldots, U_n$ with zero mean and finite variance, if there exists a positive constant $M$ such that $|U_i| \leq M$ for all $i = 1, \ldots, n$, then the following inequality holds:
\begin{equation*}
    P\left(\left|\sum_{i=1}^n U_i \right|\geq t \right) \leq 2 \exp \left(- \frac{t^2}{2Mt/3 + 2 \sum_{i=1}^n \text{Var}(U_i)} \right).
\end{equation*}

Using Bernstein's inequality, we obtain
\begin{align*}
    &P(T \geq t)  \\
    &= P\left(\max_j \left|\sum_{i=1}^n g_j(\boldsymbol{X}_i, \boldsymbol{X}_i^\prime) / (r_j F) \right| \geq t \right) \leq \sum_{j=1}^{\mathcal{N}_n}  P\left(\left|\sum_{i=1}^n g_j(\boldsymbol{X}_i, \boldsymbol{X}_i^\prime) / (r_j F) \right| \geq t \right) \land 1\\
    &\leq 2\sum_{j=1}^{\mathcal{N}_n}  \exp \left(- \frac{t^2}{\frac{2}{3} \frac{4}{r_j} t + \frac{8n}{F^2}}\right) \land 1 \leq  2 \sum_{j=1}^{\mathcal{N}_n} \exp \left(- \frac{t^2}{\frac{8}{3 \sqrt{n^{-1}\log \mathcal{N}_n}}t + \frac{8n}{F^2}} \right) \land 1\\
    &= 2 \mathcal{N}_n  \exp \left(- \frac{t^2}{\frac{8}{3 \sqrt{n^{-1}\log \mathcal{N}_n}}t + \frac{8n}{F^2}} \right) \land 1.
\end{align*}
Furthermore, when $t \geq 6 \sqrt{n \log \mathcal{N}_n}$, using the fact that $F \geq 1$, we have
\begin{equation*}
    P(T \geq t) \leq 2 \mathcal{N}_n  \exp \left(- \frac{3\sqrt{\log \mathcal{N}_n}}{16\sqrt{n}}t \right) \land 1.
\end{equation*}
Therefore, we obtain
\begin{align}
    E[T] &= \int_0^\infty P(T \geq t) dt \notag \\
    &= \int_0^{6\sqrt{n \log \mathcal{N}_n}} P(T \geq t) dt + \int_{6\sqrt{n \log \mathcal{N}_n}}^\infty P(T \geq t) dt \notag\\
    &\leq 6\sqrt{n \log \mathcal{N}_n} + \int_{6\sqrt{n \log \mathcal{N}_n}}^\infty 2 \mathcal{N}_n \exp \left(- \frac{3 \sqrt{\log \mathcal{N}_n}}{16\sqrt{n}}t \right) dt \notag\\
    &= 6\sqrt{n \log \mathcal{N}_n}  + 2 \mathcal{N}_n \frac{16\sqrt{n}}{3 \sqrt{\log \mathcal{N}_n}} \exp\left(- \frac{18\sqrt{n} \log \mathcal{N}_n}{16\sqrt{n}} \right) \notag\\
    &\leq  6\sqrt{n \log \mathcal{N}_n} + \frac{32}{3} \sqrt{\frac{n}{\log \mathcal{N}_n}}. \label{eq:E[T]}
\end{align}
Similarly, 
\begin{align}
    E[T^2] &= \int_0^\infty P(T^2 \geq u) du = \int_0^\infty P(T \geq \sqrt{u}) du \notag \\
    &= \int_0^{36n\log \mathcal{N}_n} P(T \geq \sqrt{u}) du + \int_{36n\log \mathcal{N}_n}^\infty P(T \geq \sqrt{u}) du \notag\\
    &\leq 36n \log \mathcal{N}_n + \int_{36n\log \mathcal{N}_n}^\infty 2 \mathcal{N}_n \exp \left(- \frac{3 \sqrt{u} \sqrt{\log \mathcal{N}_n}}{16\sqrt{n}} \right) du \notag\\
    &= 36n \log \mathcal{N}_n + 2 \mathcal{N}_n \underset{\leq \frac{1}{2}}{\underbrace{\frac{\frac{18}{16} \log \mathcal{N}_n +1}{9 \log \mathcal{N}_n}}} 2^8 n \exp\left(- \frac{18}{16} \log \mathcal{N}_n \right) \notag\\
    &\leq 36n \log \mathcal{N}_n + 2^8 n \label{eq:E[T^2]}
\end{align}
where the fourth inequality uses $\int_{b^2}^\infty e^{- \sqrt{u} a} du = 2 \int_b^\infty s e^{-sa} ds = 2(ba+1)e^{-ba}/a^2$. 

Substituting  (\ref{eq:E[T]}) and (\ref{eq:E[T^2]}) into (\ref{eq:difference estimation error}) and using $\log \mathcal{N}_n \leq n$, we obtain 
\begin{align}
    &|R(\hat{f}_n, f^{\circ}) - \hat{R}_n(\hat{f}_n, f^{\circ})| \notag\\
    &\leq \frac{F^2}{n} R(\hat{f}_n, f^{\circ})^{1/2} \left(36n \log \mathcal{N}_n + 2^8 n \right)^{1/2} \notag \\
    &\quad \quad \quad + \frac{F^2}{n}\left(\sqrt{\frac{\log \mathcal{N}_n}{n}} + \delta\right) \left(6\sqrt{n \log \mathcal{N}_n} + \frac{32}{3} \sqrt{\frac{n}{\log \mathcal{N}_n}} \right) + 9 \delta F \notag\\
    &\leq \frac{F^2}{n} R(\hat{f}_n, f^{\circ})^{1/2} \left(36n \log \mathcal{N}_n + 2^8 n \right)^{1/2} + \frac{F^2}{n}(6 \log \mathcal{N}_n +11) + 26\delta F^2. \label{eq:difference estimation error 2}
\end{align}
For (\ref{eq:difference estimation error 2}), we define $a,b,c$ and $d$ in \Cref{lemma : both inequality} as follows:
\begin{align*}
    & a = R(\hat{f}_n, f^{\circ}),\quad  b= \hat{R}_n(\hat{f}_n,f^{\circ}),\\
    & c= F^2\frac{(9n \log \mathcal{N}_n + 64n)^{\frac{1}{2}}}{n} \quad \text{and} \ d= F^2\frac{(6\log \mathcal{N}_n +11)}{n} + 26\delta F^2.
\end{align*}
In this case, it is satisfied that $|a - b| \leq 2\sqrt{a} c +d$.
By \Cref{lemma : both inequality}, for any $\epsilon, \delta \in (0,1]$,
\begin{align*}
    &R(\hat{f}_n, f^{\circ}) \\
    &\leq (1 + \epsilon)\left(\hat{R}_n(\hat{f}_n, f^{\circ})+ F^2 \frac{6 \log \mathcal{N}_n+11}{n} + 26\delta F^2 \right) + \frac{(1+\epsilon)^2}{\epsilon} F^4 \frac{9n \log \mathcal{N}_n + 64n}{n^2}.
\end{align*}
For the right hand side, using $\epsilon/(1+\epsilon) \leq 1/2$, we get
\begin{align*}
    &(1 + \epsilon)\left(\hat{R}_n(\hat{f}_n, f^{\circ})+ F^2 \frac{6 \log \mathcal{N}_n+11}{n} + 26\delta F^2 \right) + \frac{(1+\epsilon)^2}{\epsilon} F^4 \frac{9n \log \mathcal{N}_n + 64n}{n^2}  \\
    &= (1 + \epsilon)\left(\hat{R}_n(\hat{f}_n, f^{\circ})+ \frac{\epsilon}{1+\epsilon} \cdot \frac{1+\epsilon}{\epsilon} F^2 \frac{6 \log \mathcal{N}_n+11}{n} + 26\delta F^2 \right)  \\
    &\quad \quad \quad + \frac{(1+\epsilon)^2}{\epsilon} F^4 \frac{9n \log \mathcal{N}_n + 64n}{n^2}  \\
    &\leq  (1 + \epsilon)\left(\hat{R}_n(\hat{f}_n, f^{\circ})+ (1+\epsilon) F^4 \frac{12 \log \mathcal{N}_n+ 70}{n \epsilon} + 26\delta F^2 \right). 
\end{align*}
(I) is established.\\
\textit{(II)} : Define an arbitrary estimator $\tilde{f}$ in $\Phi$, and take $j^{\prime}$ such that $|\tilde{f} - f_{j^{\prime}}|_{\infty} \leq \delta$. Then, we have
\begin{align}
    &\left|E\left\lbrack \sum_{i=1}^n \int \xi_i(t)(\tilde{f}(\boldsymbol{X}_i, t) - f_{j^{\prime}}(\boldsymbol{X}_i,t)) dt\right\rbrack \right| \notag \\
    &\leq E \left\lbrack \sum_{i=1}^n \int |\xi_i(t)| \cdot |\tilde{f}(\boldsymbol{X}_i,t) - f_{j^{\prime}}(\boldsymbol{X}_i, t)| dt \right\rbrack \leq \delta E \left\lbrack \sum_{i=1}^n \int |\xi_i(t)| dt \right\rbrack \notag\\
    &=\delta \sum_{i=1}^n E \left\lbrack \int |\xi_i(t)| dt \right\rbrack  = n\delta E\left\lbrack \int |\xi_1(t)| dt\right\rbrack = n\delta K, \label{eq:(II)-1}
\end{align}
where $K=E[\int |\xi_1(t)| dt] < \infty$.

By using Fubini's theorem, we have $E[\int \xi_i(t) f^{\circ}(\boldsymbol{X}_i, t) dt] = \int E[\xi_i(t) f^\circ (\boldsymbol{X}_i, t)] dt = \int E[E[\xi_i(t) f^\circ (\boldsymbol{X}_i, t) | \boldsymbol{X}_i]]dt = 0$. Using  (\ref{eq:(II)-1}), we find
\begin{align}
    &\left|E \left\lbrack \frac{2}{n} \sum_{i=1}^n \int_0^1 \xi_i(t) \tilde{f}(\boldsymbol{X}_i, t) dt \right\rbrack \right| \notag\\
    &= \left| E \left\lbrack \frac{2}{n} \sum_{i=1}^n \int_0^1 \xi_i(t) (\tilde{f}(\boldsymbol{X}_i, t) - f^{\circ}(\boldsymbol{X}_i, t)) dt \right\rbrack \right| \notag\\
    &= \left| E \left\lbrack \frac{2}{n} \sum_{i=1}^n \int_0^1 \xi_i(t) (\tilde{f}(\boldsymbol{X}_i, t) - f_{j^\prime}(\boldsymbol{X}_i, t) + f_{j^\prime}(\boldsymbol{X}_i, t) - f^{\circ}(\boldsymbol{X}_i, t)) dt \right\rbrack \right| \notag\\
    &\leq \frac{2}{n} \left|E\left\lbrack \sum_{i=1}^n \int_0^1 \xi_i(t) (\tilde{f}(\boldsymbol{X}_i, t) - f_{j^\prime}(\boldsymbol{X}_i, t)) dt\right\rbrack \right| \notag \\
    &\quad \quad \quad \quad + \left| E \left\lbrack \frac{2}{n} \sum_{i=1}^n \int_0^1 \xi_i(t) (f_{j^\prime}(\boldsymbol{X}_i, t) - f^{\circ}(\boldsymbol{X}_i, t)) dt \right\rbrack \right| \notag\\
    &\leq 2 \delta K + \frac{2}{\sqrt{n}} E \left\lbrack (\|\tilde{f} - f^{\circ}\|_n + \delta) |\eta_{j^\prime}|\right\rbrack,  \label{eq:(II)-2}
\end{align}
where we define
\begin{equation*}
    \eta_j := \frac{\sum_{i=1}^n \int_0^1 \xi_i(t) (f_j(\boldsymbol{X}_i, t) - f^{\circ}(\boldsymbol{X}_i,t)) dt}{\sqrt{n} \|f_j - f^{\circ}\|_n}. 
\end{equation*}
The last inequality in (\ref{eq:(II)-2}) is obtained using the fact that $\|f_{j^{\prime}} - f^{\circ}\|_n \leq \|f_{j^{\prime}} - \tilde{f}\|_n + \|\tilde{f} - f^{\circ}\|_n \leq \|\tilde{f} - f^{\circ}\|_n + \delta$.

Since $\eta_{j^\prime} \left.\right| \{\boldsymbol{X}_i\}_{i=1}^n$ is a Sub-Gaussian with mean $0$ and variance proxy $\sigma^2$, by \Cref{lemma : schmidt LEMMA C.1}, we obtain $E[\eta_{j^{\prime}}^2] \leq E[\max_j \eta_j^2] \leq 2\sigma^2 \log \mathcal{N}_n + 4\sigma^2$. Using Cauchy-Schwarz,
\begin{align*}
    &E\left\lbrack(\|\tilde{f} - f^{\circ}\|_n + \delta) |\eta_{j^{\prime}}| \bigl.\bigr| \{\boldsymbol{X}_i\}_{i=1}^n \right\rbrack \\
    &\leq \left(E\left\lbrack \|\tilde{f} - f^{\circ}\|_n^2 \bigl.\bigr| \{\boldsymbol{X}_i\}_{i=1}^n\right\rbrack\right)^{1/2} \left(E\left\lbrack |\eta_{j^{\prime}}|^2 \bigl.\bigr| \{\boldsymbol{X}_i\}_{i=1}^n \right\rbrack \right)^{1/2} + \delta \left( E \left\lbrack |\eta_{j^{\prime}}|^2 \bigl.\bigr|\{\boldsymbol{X}_i\}_{i=1}^n \right\rbrack \right) \\
    &\leq \left\lbrace \left(E\left\lbrack \|\tilde{f} - f^{\circ}\|_n^2 \bigl.\bigr| \{\boldsymbol{X}_i\}_{i=1}^n \right\rbrack\right)^{1/2} + \delta \right\rbrace \sqrt{2\sigma^2 \log \mathcal{N}_n + 4\sigma^2} \\
    & = \left(\|\tilde{f} - f^{\circ}\|_n + \delta \right) \sqrt{2\sigma^2 \log \mathcal{N}_n + 4\sigma^2}.
\end{align*}
Taking the expectation with respect to $\{\boldsymbol{X}_i\}_{i=1}^n$,
\begin{align}
    &E \left\lbrack (\|\tilde{f} - f^{\circ}\|_n + \delta) |\eta_{j^{\prime}}|\right\rbrack \notag\\
    &= E \left\lbrack E \left\lbrack (\|\tilde{f} - f^{\circ}\|_n + \delta) |\eta_{j^{\prime}}| \bigl.\bigr| \{\boldsymbol{X}_i\}_{i=1}^n \right\rbrack \right\rbrack \notag \\
    &\leq E \left\lbrack \left( \|\tilde{f} - f^{\circ}\|_n + \delta \right) \sqrt{2\sigma^2 \log \mathcal{N}_n + 4\sigma^2}\right\rbrack \notag\\
    &\leq \left(E \left\lbrack \|\tilde{f} - f^{\circ}\|_n^2\right\rbrack\right)^{1/2} \sqrt{2\sigma^2 \log \mathcal{N}_n + 4\sigma^2} + \delta \sqrt{2\sigma^2 \log \mathcal{N}_n + 4\sigma^2} \notag \\
    &= \left( \hat{R}_n^{1/2} (\tilde{f}, f^{\circ}) + \delta\right) \sqrt{2\sigma^2 \log \mathcal{N}_n + 4\sigma^2}. \label{eq:(II)-3}
\end{align}
Substituting the result of (\ref{eq:(II)-3}) into (\ref{eq:(II)-2}) and using $\log \mathcal{N}_n \leq n$, we have
\begin{align*}
    &\left|E\left\lbrack \frac{2}{n} \sum_{i=1}^n \int \xi_i(t) \tilde{f}(\boldsymbol{X}_i, t) dt\right\rbrack \right| \\
    & \leq 2 \delta K + \frac{2}{\sqrt{n}} E \left\lbrack (\|\tilde{f} - f^{\circ}\|_n + \delta)|\eta_{j^\prime}| \right\rbrack \\
    &\leq 2 \delta K + \frac{2}{\sqrt{n}} (\hat{R}_n^{1/2}(\tilde{f}, f^{\circ}) + \delta) \sqrt{2\sigma^2 \log \mathcal{N}_n + 4\sigma^2} \\
    &= 2 \delta K + 2 \sqrt{\frac{\hat{R}_n(\tilde{f}, f^{\circ}) (2\sigma^2 \log \mathcal{N}_n + 4\sigma^2)}{n}} + 2 \delta \underset{\leq 3\sigma}{\underbrace{\sqrt{\frac{2\sigma^2 \log \mathcal{N}_n + 4\sigma^2}{n}}}} \\
   &\leq 2 \sqrt{\frac{\hat{R}_n(\tilde{f}, f^{\circ}) (2\sigma^2 \log \mathcal{N}_n + 4\sigma^2)}{n}}  + 2 \delta K + 6\delta \sigma.
\end{align*}
Thus, (II) is proven.

\textit{(III)} : For any $f \in \Phi$, by the definition of $\Delta_n$, we find
\begin{equation*}
    E\left\lbrack \frac{1}{n} \sum_{i=1}^n \int (Y_i(t) - \hat{f}(\boldsymbol{X}_i, t))^2 dt\right\rbrack \leq E \left\lbrack \frac{1}{n} \sum_{i=1}^n \int (Y_i(t) - f(\boldsymbol{X}_i, t))^2 dt \right\rbrack + \Delta_n. 
\end{equation*}
Since $\boldsymbol{X}_i$ and $\boldsymbol{X}$ are independently and identically distributed and $f, f^{\circ}$ do not depend on them, it holds that $E[\|f - f^{\circ}\|_n^2] = E[\frac{1}{n} \sum_{i=1}^n \int (f(\boldsymbol{X}_i, t) - f^{\circ}(\boldsymbol{X}_i, t))^2 dt] = E[\int (f(\boldsymbol{X}, t) - f^{\circ}(\boldsymbol{X}, t))^2 dt]$. 
Using $E[\int \xi_i(t) f(\boldsymbol{X}_i, t) dt] = 0$, we have
\begin{equation*}
    \hat{R}_n(\hat{f}, f^{\circ}) \leq E \left\lbrack \int (f(\boldsymbol{X},t) - f^{\circ}(\boldsymbol{X},t))^2 dt\right\rbrack + E \left\lbrack \frac{2}{n} \sum_{i=1}^n \int \xi_i(t) \hat{f}_n(\boldsymbol{X}_i, t) dt\right\rbrack + \Delta_n.
\end{equation*}
From the result of (II), we have
\begin{align*}
    \hat{R}_n(\hat{f}, f^{\circ}) &\leq E \left\lbrack \int (f(\boldsymbol{X},t) - f^{\circ}(\boldsymbol{X},t))^2 dt\right\rbrack  \\
    &+ 2 \sqrt{\frac{\hat{R}_n(\hat{f}_n , f^{\circ}) (2\sigma^2 \log \mathcal{N}_n + 4\sigma^2)}{n}} + 2\delta K + 6\delta \sigma + \Delta_n. 
\end{align*}
Here, we define $a, b, c,$ and $d$ in \Cref{lemma : both inequality} as follows:
\begin{align*}
    & a= \hat{R}_n(\hat{f}_n, f^{\circ}) , \  b = 0, \ c= \sqrt{\frac{2\sigma^2 \log \mathcal{N}_n + 4\sigma^2}{n}} , \\ 
    & \text{and} \ d = E \left\lbrack \int (f(\boldsymbol{X},t) - f^{\circ}(\boldsymbol{X},t))^2 dt\right\rbrack + 2\delta K + 6\delta \sigma + \Delta_n.
\end{align*}
In this case, it is satisfied that $|a - b| \leq 2 \sqrt{a} c + d$, and using $1+\epsilon \leq 2$ and $F \geq 1$, we obtain
\begin{align*}
    \hat{R}_n(\hat{f}, f^{\circ}) &\leq (1 + \epsilon)\left(E\left\lbrack\int (f(\boldsymbol{X}, t) - f^{\circ}(\boldsymbol{X}, t))^2 dt\right\rbrack + 2\delta K + 6\delta \sigma + \Delta_n \right)  \\
    &\hspace{5cm} + \frac{(1+\epsilon)^2}{\epsilon} \frac{2\sigma^2 \log \mathcal{N}_n + 4\sigma^2}{n} \\
    &\leq (1 + \epsilon)\left(E\left\lbrack\int (f(\boldsymbol{X}, t) - f^{\circ}(\boldsymbol{X}, t))^2 dt\right\rbrack + 2\delta K + 6\delta \sigma +  \frac{4\sigma^2 \log \mathcal{N}_n + 8\sigma^2}{n\epsilon} +\Delta_n  \right) \\
    &\leq (1 + \epsilon)\left(E\left\lbrack\int (f(\boldsymbol{X}, t) - f^{\circ}(\boldsymbol{X}, t))^2 dt\right\rbrack + 2\delta K + 6\delta \sigma +  F^2 \frac{4\sigma^2 \log \mathcal{N}_n + 8\sigma^2}{n\epsilon} +\Delta_n  \right).
\end{align*}
Since $f$ is arbitrary, we take the infimum over $f \in \Phi$ on the right hand side, and we obtain 
\begin{align*}
    \hat{R}_n(\hat{f}_n, f^{\circ}) &\leq (1+\epsilon) \left( \inf_{f \in \Phi} E \left\lbrack \int (f(\boldsymbol{X},t) - f^{\circ}(\boldsymbol{X},t))^2 dt\right\rbrack + \right. \\ 
    &\hspace{5cm} \left. 2\delta K + 6\delta \sigma + F^2 \frac{4\sigma^2 \log \mathcal{N}_n + 8\sigma^2}{n \epsilon} + \Delta_n \right).
\end{align*}
This completes the proof.

\begin{lemma}{\textnormal{(\cite{schmidt2020Nonpara})}} \label{lemma : both inequality}
    Let $a, b, c, d > 0$ be such that $|a - b| \leq 2 \sqrt{a} c + d$. Then, for any $\epsilon \in (0,1]$, the following holds:
    \begin{equation*}
        a \leq (1+\epsilon) (b+d) + \frac{(1+\epsilon)^2}{\epsilon} c^2.
    \end{equation*}
\end{lemma}
\begin{proof}
By $|a - b|\leq 2\sqrt{a} c + d$, we obtain $a \leq 2\sqrt{a}c + b + d$.
For the upper bound, by the arithmetic-geometric mean inequality, we have
\begin{align*}
    &a \leq 2 \sqrt{a} c + b + d \leq \frac{\epsilon}{1+\epsilon} a + \frac{1+\epsilon}{\epsilon} c^2 + b + d.
\end{align*}
Rearranging the terms for $a$, we have
\begin{align*}
    &\left(1 - \frac{\epsilon}{1+\epsilon} \right) a \leq \frac{1+\epsilon}{\epsilon}c^2 + b+ d.
\end{align*}
By multiplying both sides by $1+\epsilon$, we obtain
\begin{equation*}
    a \leq (1+\epsilon)(b+d) + \frac{(1+\epsilon)^2}{\epsilon} c^2.
\end{equation*} 
\end{proof}

\begin{lemma}{}\label{lemma : schmidt LEMMA C.1} 
    Let $\eta_j, j=1,\ldots, M$, be i.i.d. Sub-Gaussian with mean $0$ and variance proxy $\sigma^2$. Then, it holds that $E[\max_{j=1,\ldots, M} \eta_j^2] \leq 2 \sigma^2 \log M + 4 \sigma^2$.
\end{lemma}
\begin{proof}
    Let $Z = \max_{j=1, \ldots, M} \eta_j^2$. Since $P(|\eta_j| \geq t) \leq 2 \exp \left(- \frac{t^2}{2 \sigma^2} \right)$ for $t > 0$,
\begin{align*}
    E[Z] &= \int_0^\infty P(Z \geq t) dt = \int_0^T P(Z \geq t) dt + \int_T^\infty P(Z \geq t) dt \\
    &\leq T + \int_T^\infty P(Z \geq t) dt \leq T \leq T + M \int_T^\infty P(\eta_1^2 \geq t) dt \\
    &\leq T + 2M \int_T^\infty P(\eta_1^2 \geq t) dt \leq T + 2M \int_T^\infty e^{- \frac{t}{2 \sigma^2}} dt = T + 4M\sigma^2 e^{- \frac{T}{2 \sigma^2}}.
\end{align*}
By taking $T = 2 \sigma^2 \log M$, we obtain $E[Z] \leq 2 \sigma^2 \log M + 4\sigma^2$.    
\end{proof}

\newpage

\section{Additional infomation in numerical experiments} \label{sec: Appendix C}
\leavevmode

\subsection*{Scenario 1}
\leavevmode

\begin{figure}[h]
    \centering
    \includegraphics[width=\linewidth]{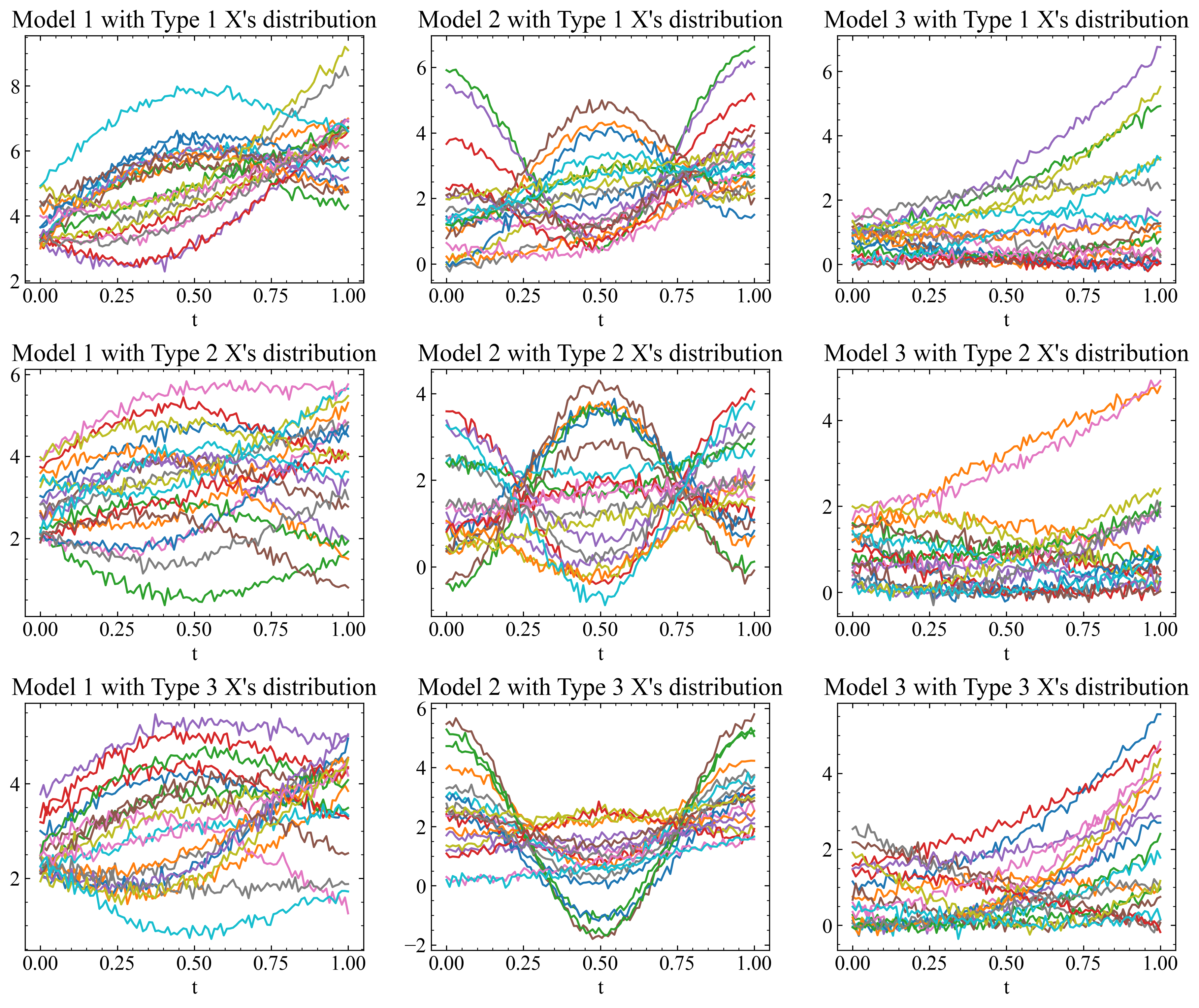}
    \caption{Sample curves of $Y(t)$ for the 9 settings which are nine different combinations of three models and three types of joint distributions of $(X_1, X_2, X_3)$ in Scenario 1.}
    \label{fig:Simulation 1}
\end{figure}

\clearpage

\begin{table}[h]
    \centering
    \caption{Averages (and standard deviations) of the MISPEs from 50 replicates for all settings and methods in Scenario 1 ($N_{\text{train}} = 1000$). 
    For the \textit{FOS-DNN}, we set $W=32$, $L=6$, and $\alpha = 10^{-5}$. For \textit{FUA}, \textit{famm.nl} and \textit{famm.ad}, the parameter settings in Luo and Qi (2023) are used.}
    \label{tab:2}
    \resizebox{\textwidth}{!}{
    \begin{tabular}{|c|c|c|c|c|c|c|}
  \hline
   \textbf{X-type} &  \textbf{Model} &  \textbf{FOS-DNN} &  \textbf{FUA} &  \textbf{famm.nl} &  \textbf{famm.ad} &  \textbf{fosr} \\
  \hline \hline 
   1 &  1 &  0.011 (0.000) &  \textbf{0.010 (0.000)} &  0.011 (0.000) &  \textbf{0.010 (0.000)} &  0.124 (0.004) \\
   &  2 &  0.013 (0.001) &  \textbf{0.011 (0.000)} &  0.013 (0.000) &  0.615 (0.027) &  1.015 (0.024)\\
   &  3 &  \textbf{0.011 (0.000)} & \textbf{0.011 (0.000)} &  \textbf{0.011 (0.000)} &  0.810 (0.031) &  1.018 (0.038)  \\ \hline
   2 &  1 &  0.012 (0.001) &  \textbf{0.011 (0.000)} &  0.016 (0.003) &  \textbf{0.011 (0.001)} &  0.452 (0.019)\\
   &  2 &  0.017 (0.001) &  0.017 (0.007) &  \textbf{0.014 (0.001)} &  0.758 (0.034) &  1.015 (0.026)\\
   &  3 &  \textbf{0.012 (0.000)} & \textbf{0.012 (0.001)} &  \textbf{0.012 (0.001)} &   0.795 (0.030) &  1.009 (0.031)\\
  \hline
   3 &  1 &  0.012 (0.001) & \textbf{0.011 (0.001)} &  0.015 (0.002) &  0.011 (0.001) &  0.379 (0.015) \\
  &  2 &  0.016 (0.001) & \textbf{0.014 (0.001)} &  0.015 (0.001) &  0.471 (0.030)  &  1.016 (0.020) \\
  &  3 &  \textbf{0.012 (0.000)} & 0.013 (0.001) &  0.014 (0.001) &  0.806 (0.031)  &  1.020 (0.032) \\ \hline 
\end{tabular}
}
\end{table}

\clearpage

\begin{table}[t]
    \centering
    \caption{Averages (Standard deviations) of 3-fold Cross-Validation in Model 2 of Scenario 1 ($N_{\text{train}} = 200$).}
    \begin{tabular}{|c|c|c|l||c|c|c|l||c|c|c|l|}
    \hline \hline
        W & L & $\alpha$ & mean (std) & W & L & $\alpha$ & mean (std) & W & L & $\alpha$ & mean (std) \\ \hline \hline
         &  & $10^{-9}$ & 0.311 (0.066) &  &  & $10^{-9}$ & 0.112 (0.036) &  &  & $10^{-9}$ & 0.072 (0.025) \\ 
         &  & $10^{-7}$ & 0.425 (0.097) &  &  & $10^{-7}$ & 0.110 (0.042) &  &  & $10^{-7}$ & 0.066 (0.014) \\ 
         8 & 5 & $10^{-5}$ & 0.393 (0.114) & 16 & 5 & $10^{-5}$ & 0.097 (0.025) & 32 & 5 & $10^{-5}$ & 0.068 (0.021) \\ 
         &  & $10^{-3}$ & 0.234 (0.081) &  &  & $10^{-3}$ & 0.110 (0.040) &  &  & $10^{-3}$ & 0.064 (0.029) \\ 
         &  & $10^{-1}$ & 1.053 (0.096) &  &  & $10^{-1}$ & 1.052 (0.096) &  &  & $10^{-1}$ & 1.052 (0.096) \\ \hdashline
         &  & $10^{-9}$ & 0.388 (0.062) &  &  & $10^{-9}$ & 0.089 (0.030) &  &  & $10^{-9}$ & 0.110 (0.053) \\ 
         &  & $10^{-7}$ & 0.307 (0.216) &  & & $10^{-7}$ & 0.119 (0.032) &  &  & $10^{-7}$ & 0.075 (0.028) \\ 
         8 & 6 & $10^{-5}$ & 0.225 (0.066) & 16 & 6 & $10^{-5}$ & 0.113 (0.051) & 32 & 6 & $10^{-5}$ & 0.092 (0.048) \\
         &  & $10^{-3}$ & 0.180 (0.108) &  &  & $10^{-3}$ & 0.109 (0.028) &  &  & $10^{-3}$ & 0.071 (0.046) \\
         &  & $10^{-1}$ & 1.052 (0.096) &  &  & $10^{-1}$ & 1.052 (0.097) &  &  & $10^{-1}$ & 1.052 (0.096) \\ \hdashline
         &  & $10^{-9}$ & 0.476 (0.356) &  &  & $10^{-9}$ & 0.148 (0.073) &  &  & $10^{-9}$ & 0.087 (0.033) \\ 
         &  & $10^{-7}$ & 0.264 (0.148) &  &  & $10^{-7}$ & 0.152 (0.045) &  &  & $10^{-7}$ & 0.072 (0.023) \\ 
         8& 7 & $10^{-5}$ & 0.249 (0.137) & 16 & 7 & $10^{-5}$ & 0.108 (0.047) & 32 & 7 & $10^{-5}$ & 0.087 (0.015) \\ 
         &  & $10^{-3}$ & 0.206 (0.106) &  &  & $10^{-3}$ & 0.101 (0.025) &  &  & $10^{-3}$ & 0.057 (0.018) \\ 
         &  & $10^{-1}$ & 1.052 (0.096) &  &  & $10^{-1}$ & 1.052 (0.096) &  &  & $10^{-1}$ & 1.053 (0.097) \\ \hline 
    \end{tabular}
\end{table}

\begin{table}[t]
    \centering
    \caption{Averages (Standard deviations) of 3-fold Cross-Validation in Model 2 of Scenario 1 ($N_{\text{train}} = 1000$).}
    \begin{tabular}{|c|c|c|l||c|c|c|l||c|c|c|l|}
    \hline \hline
        W & L & $\alpha$ & mean (std) & W & L & $\alpha$ & mean (std) & W & L & $\alpha$ & mean (std) \\ \hline \hline
 &  & $10^{-9}$ & 0.087 (0.013) &  &  & $10^{-9}$ & 0.029 (0.002) &  &  & $10^{-9}$ & 0.019 (0.001) \\
 &  & $10^{-7}$ & 0.129 (0.031) &  &  & $10^{-7}$ & 0.036 (0.006) &  &  & $10^{-7}$ & 0.019 (0.001) \\
8 & 5 & $10^{-5}$ & 0.086 (0.014) & 16 & 5 & $10^{-5}$ & 0.035 (0.004) & 32 & 5 & $10^{-5}$ & 0.020 (0.002) \\
 &  & $10^{-3}$ & 0.113 (0.015) &  &  & $10^{-3}$ & 0.038 (0.002) &  &  & $10^{-3}$ & 0.024 (0.001) \\
 &  & $10^{-1}$ & 1.087 (0.036) &  &  & $10^{-1}$ & 1.088 (0.036) &  &  & $10^{-1}$ & 1.087 (0.038) \\ \hdashline
 &  & $10^{-9}$ & 0.091 (0.022) &  &  & $10^{-9}$ & 0.028 (0.007) &  &  & $10^{-9}$ & 0.018 (0.000) \\
 &  & $10^{-7}$ & 0.105 (0.028) &  &  & $10^{-7}$ & 0.026 (0.001) &  &  & $10^{-7}$ & 0.018 (0.000) \\
8 & 6 & $10^{-5}$ & 0.089 (0.004) & 16 & 6 & $10^{-5}$ & 0.025 (0.001) & 32 & 6 & $10^{-5}$ & 0.019 (0.002) \\
 &  & $10^{-3}$ & 0.090 (0.018) &  &  & $10^{-3}$ & 0.031 (0.005) &  &  & $10^{-3}$ & 0.022 (0.001) \\
 &  & $10^{-1}$ & 1.087 (0.036) &  &  & $10^{-1}$ & 1.088 (0.037) &  &  & $10^{-1}$ & 1.085 (0.036) \\ \hdashline
 &  & $10^{-9}$ & 0.071 (0.013) &  &  & $10^{-9}$ & 0.024 (0.003) &  &  & $10^{-9}$ & 0.017 (0.001) \\
 &  & $10^{-7}$ & 0.086 (0.027) &  &  & $10^{-7}$ & 0.027 (0.002) &  &  & $10^{-7}$ & 0.019 (0.001) \\
8 & 7 & $10^{-5}$ & 0.097 (0.015) & 16 & 7 & $10^{-5}$ & 0.023 (0.001) & 32 & 7 & $10^{-5}$ & 0.018 (0.001) \\
 &  & $10^{-3}$ & 0.088 (0.015) &  &  & $10^{-3}$ & 0.033 (0.005) &  &  & $10^{-3}$ & 0.020 (0.001) \\
 &  & $10^{-1}$ & 1.087 (0.036) &  &  & $10^{-1}$ & 1.088 (0.035) &  &  & $10^{-1}$ & 1.087 (0.036) \\
    \hline
    \end{tabular}
\end{table}

\clearpage

\subsection*{Scenario 1-A}
\leavevmode

\begin{align*}
    \textbf{Model 1} : \ &f^\circ(X_1,X_2, X_3, X_4, X_5, t) = \\
    & \quad\sin (X_1 + t) + \frac{(X_2 + t)^2}{X_2^2 + 1} + \frac{e^{X_3 t}}{1+e^{X_3}} + \log (1 + X_4^2 + t^2) + X_5(1-t) \cos (2 \pi t), \\
    \textbf{Model 2} : \ & f^\circ(X_1,X_2, X_3, X_4, X_5, t) = \cos(X_1 - 2 X_2 + X_3) t^2 + \sin (X_4 + X_5) e^t ,\\
    \textbf{Model 3} :  \ & f^\circ(X_1,X_2, X_3, X_4, X_5, t) = \frac{X_1 - X_2 t + 2 X_3 t^2}{1 + 2 X_4^2 + 4 X_5^2 \cos^2 (2 \pi t)}.
\end{align*}

\begin{figure}[h]
    \centering
    \includegraphics[width=\linewidth]{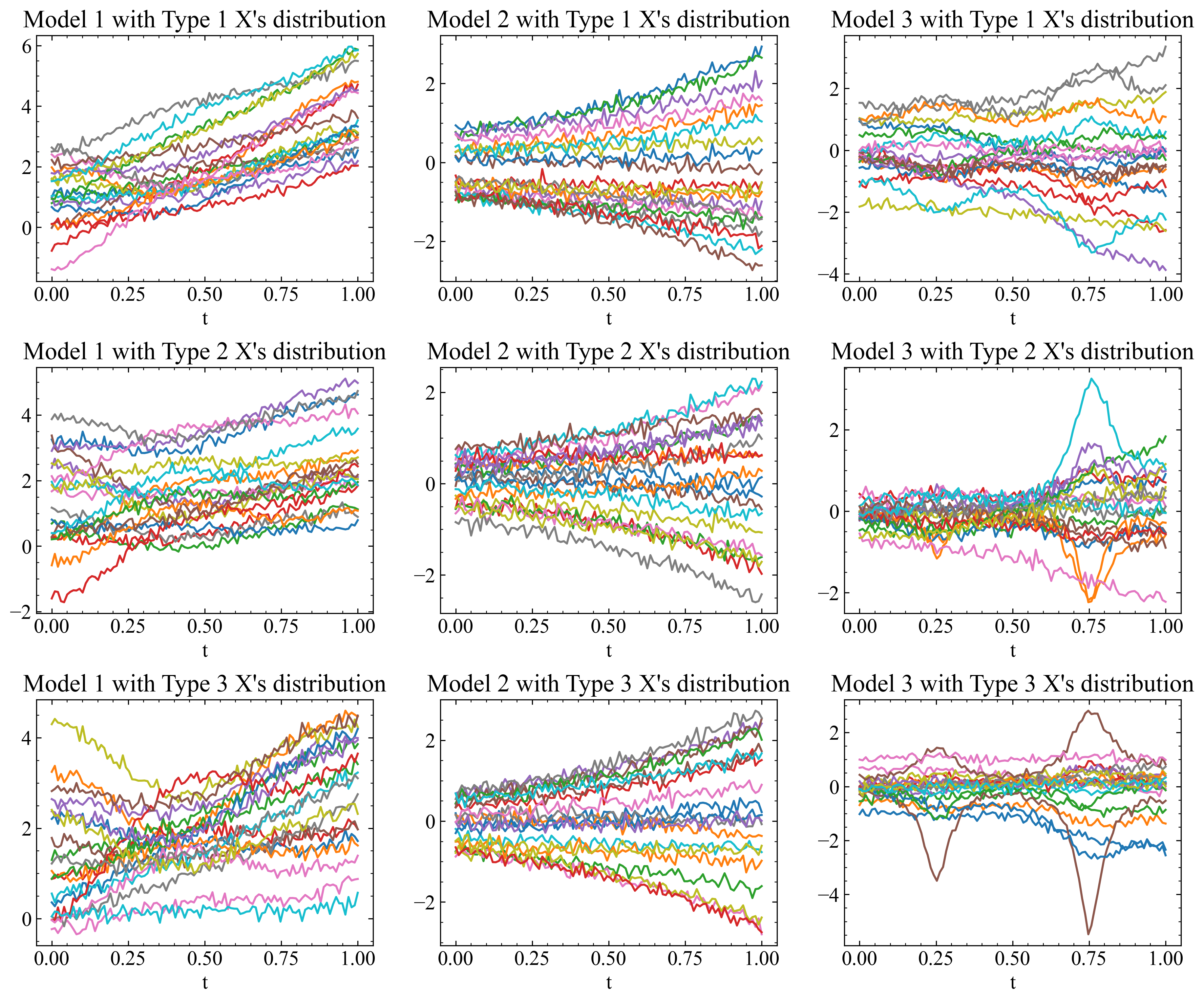}
    \caption{Sample curves of $Y(t)$ for the 9 settings which are nine different combinations of three models and three types of joint distributions of $(X_1, X_2, X_3)$ in Scenario 1-A. }
    \label{fig:Simulation 2}
\end{figure}

\clearpage

\begin{table}[h]
    \centering
    \caption{Averages (and standard deviations) of the MISPEs from 100 replicates for all settings and methods in Scenario 1-A ($N_{\text{train}} = 200$). 
    For the \textit{FOS-DNN}, we set $W=32$, $L=6$, and $\alpha = 10^{-3}$. For \textit{FUA} and \textit{famm.ad}, the parameter settings in \cite{luo2023nonlinear} are used.}
    \label{tab:Simulation 2 200}
\begin{tabular}{|c|c|c|c|c|c|c|}
  \hline
   \textbf{X-type} &  \textbf{Model} &  \textbf{FOS-DNN} &  \textbf{FUA} &  \textbf{famm.ad} &  \textbf{fosr} \\
  \hline \hline 
   1 &  1 &  0.016 (0.001) &  \textbf{0.012 (0.000)} &  \textbf{0.012 (0.000)} &  0.128 (0.005) \\
   &  2 &  0.013 (0.001) &  \textbf{0.011 (0.000)} &  0.090 (0.005) &  0.097 (0.004) \\
   &  3 & 0.058 (0.006) &  \textbf{0.012 (0.001)} &  0.172 (0.019) &  0.143(0.006)  \\ \hline
   2 &  1 &  0.037 (0.006) & \textbf{0.017 (0.002)} &  0.014 (0.006) &  0.406 (0.021)\\
   &  2 &  0.035 (0.008) & \textbf{0.027 (0.025)} &  0.553 (0.067) &  0.570 (0.039) \\
   &  3 &  0.145 (0.016) & \textbf{0.014 (0.002)} &  0.406 (0.039) &   0.318 (0.014)\\
  \hline
   3 &  1 & 0.029 (0.004) &  \textbf{0.016 (0.002)} &  0.013 (0.001) &  0.367 (0.025) \\
  &  2 &  0.029 (0.007) & \textbf{0.026 (0.009)} &  0.425 (0.065) &  0.758 (0.039) \\
  &  3 &  0.124 (0.018) &  \textbf{0.014 (0.002)} &  0.362 (0.048) &  0.321 (0.016)  \\ \hline 
\end{tabular}
\end{table}

\begin{table}[h]
    \centering
    \caption{Averages (and standard deviations) of the MISPEs from 50 replicates for all settings and methods in Scenario 1-A ($N_{\text{train}} = 1000$). 
    For the \textit{FOS-DNN}, we set $W=32$, $L=6$, and $\alpha = 10^{-5}$. For \textit{FUA} and \textit{famm.ad}, the parameter settings in \cite{luo2023nonlinear} are used.}
    \label{tab:Simulation 2 1000}
\begin{tabular}{|c|c|c|c|c|c|c|}
  \hline
   \textbf{X-type} &  \textbf{Model} &  \textbf{FOS-DNN} &  \textbf{FUA} &  \textbf{famm.ad} &  \textbf{fosr} \\
  \hline \hline 
   1 &  1 &  0.011 (0.000) & \textbf{0.010 (0.000)} &  0.012 (0.000) &  0.126 (0.005)  \\
   &  2 &  \textbf{0.010 (0.000)} &  \textbf{0.010 (0.000)} &  0.079 (0.003) &  0.095 (0.004) \\
   &  3 &  0.017 (0.001) &  \textbf{0.011 (0.000)} &  0.150 (0.007) &  0.141 (0.007)   \\ \hline
   2 &  1 &  \textbf{0.013 (0.001)} &  \textbf{0.011 (0.000)} &  0.013 (0.000) &  0.402 (0.017) \\
   &  2 &  \textbf{0.013 (0.002)} & \textbf{0.013 (0.002)} &  0.433 (0.031) &  0.555 (0.046) \\
   &  3 &  0.025 (0.003) & \textbf{0.011 (0.000)} &  0.334 (0.021) &   0.311 (0.019) \\
  \hline
   3 &  1 &  0.013 (0.001) & \textbf{0.011 (0.000)} &  0.013 (0.000) &  0.356 (0.014)  \\
  &  2 &  0.014 (0.002) & \textbf{0.013 (0.002)} &  0.323 (0.021) &  0.736 (0.041)  \\
  &  3 &  0.022 (0.003) & \textbf{0.011 (0.000)} &  0.284 (0.015) &  0.311 (0.014) \\ \hline 
\end{tabular}
\end{table}

\begin{table}[t]
    \centering
    \caption{Averages (Standard deviations) of 3-fold Cross-Validation in Model 2 of Scenario 1-A ($N_{\text{train}} = 200$).}
    \begin{tabular}{|c|c|c|l||c|c|c|l||c|c|c|l|}
    \hline \hline
        W & L & $\alpha$ & mean (std) & W & L & $\alpha$ & mean (std) & W & L & $\alpha$ & mean (std) \\ \hline \hline
 &  & $10^{-9}$ & 0.080 (0.015) &  &  & $10^{-9}$ & 0.140 (0.018) &  &  & $10^{-9}$ & 0.164 (0.017) \\
 &  & $10^{-7}$ & 0.094 (0.014) &  &  & $10^{-7}$ & 0.165 (0.048) &  &  & $10^{-7}$ & 0.171 (0.019) \\
8 & 5 & $10^{-5}$ & 0.133 (0.035) & 16 & 5 & $10^{-5}$ & 0.126 (0.040) & 32 & 5 & $10^{-5}$ & 0.165 (0.011) \\ 
 &  & $10^{-3}$ & 0.051 (0.010) &  &  & $10^{-3}$ & 0.102 (0.040) &  &  & $10^{-3}$ & 0.083 (0.013) \\
 &  & $10^{-1}$ & 1.087 (0.027) &  &  & $10^{-1}$ & 1.086 (0.026) &  &  & $10^{-1}$ & 0.656 (0.331) \\ \hdashline
 &  & $10^{-9}$ & 0.090 (0.024) &  &  & $10^{-9}$ & 0.174 (0.012) &  &  & $10^{-9}$ & 0.124 (0.036) \\
 &  & $10^{-7}$ & 0.094 (0.031) &  &  & $10^{-7}$ & 0.149 (0.035) &  &  & $10^{-7}$ & 0.150 (0.001) \\
8 & 6 & $10^{-5}$ & 0.080 (0.030) & 16 & 6 & $10^{-5}$ & 0.195 (0.066) & 32 & 6 & $10^{-5}$ & 0.144 (0.060) \\
 &  & $10^{-3}$ & 0.050 (0.011) &  &  & $10^{-3}$ & 0.068 (0.012) &  &  & $10^{-3}$ & 0.063 (0.018) \\
 &  & $10^{-1}$ & 1.087 (0.026) &  &  & $10^{-1}$ & 1.087 (0.026) &  &  & $10^{-1}$ & 1.087 (0.027) \\ \hdashline
 &  & $10^{-9}$ & 0.118 (0.003) &  &  & $10^{-9}$ & 0.126 (0.014) &  &  & $10^{-9}$ & 0.096 (0.031) \\
 &  & $10^{-7}$ & 0.139 (0.054) &  &  & $10^{-7}$ & 0.122 (0.043) &  &  & $10^{-7}$ & 0.162 (0.022) \\
8 & 7 & $10^{-5}$ & 0.097 (0.002) & 16 & 7 & $10^{-5}$ & 0.162 (0.060) & 32 & 7 & $10^{-5}$ & 0.122 (0.029) \\
 &  & $10^{-3}$ & 0.041 (0.010) &  &  & $10^{-3}$ & 0.066 (0.015) &  &  & $10^{-3}$ & 0.071 (0.017) \\
 &  & $10^{-1}$ & 1.087 (0.026) &  &  & $10^{-1}$ & 1.087 (0.027) &  &  & $10^{-1}$ & 1.087 (0.027) \\
    \hline
    \end{tabular}
\end{table}

\begin{table}[t]
    \centering
    \caption{Averages (Standard deviations) of 3-fold Cross-Validation in Model 2 of Scenario 1-A ($N_{\text{train}} = 1000$).}
    \begin{tabular}{|c|c|c|l||c|c|c|l||c|c|c|l|}
    \hline \hline
        W & L & $\alpha$ & mean (std) & W & L & $\alpha$ & mean (std) & W & L & $\alpha$ & mean (std) \\ \hline \hline
 &  & $10^{-9}$ & 0.019 (0.001) &  &  & $10^{-9}$ & 0.016 (0.001) &  &  & $10^{-9}$ & 0.014 (0.002) \\
 &  & $10^{-7}$ & 0.022 (0.004) &  &  & $10^{-7}$ & 0.014 (0.001) &  &  & $10^{-7}$ & 0.014 (0.002) \\
8 & 5 & $10^{-5}$ & 0.022 (0.005) & 16 & 5 & $10^{-5}$ & 0.017 (0.005) & 32 & 5 & $10^{-5}$ & 0.013 (0.001) \\
 &  & $10^{-3}$ & 0.025 (0.004) &  &  & $10^{-3}$ & 0.017 (0.003) &  &  & $10^{-3}$ & 0.015 (0.002) \\
 &  & $10^{-1}$ & 0.574 (0.304) &  &  & $10^{-1}$ & 0.985 (0.023) &  &  & $10^{-1}$ & 0.574 (0.304) \\ \hdashline
 &  & $10^{-9}$ & 0.017 (0.003) &  &  & $10^{-9}$ & 0.014 (0.002) &  &  & $10^{-9}$ & 0.014 (0.001) \\
 &  & $10^{-7}$ & 0.019 (0.005) &  &  & $10^{-7}$ & 0.013 (0.001) &  &  & $10^{-7}$ & 0.015 (0.002) \\
8 & 6 & $10^{-5}$ & 0.018 (0.002) & 16 & 6 & $10^{-5}$ & 0.014 (0.000) & 32 & 6 & $10^{-5}$ & 0.014 (0.002) \\
 &  & $10^{-3}$ & 0.023 (0.006) &  &  & $10^{-3}$ & 0.017 (0.004) &  &  & $10^{-3}$ & 0.015 (0.003) \\ 
 &  & $10^{-1}$ & 0.985 (0.022) &  &  & $10^{-1}$ & 0.985 (0.022) &  &  & $10^{-1}$ & 0.985 (0.023) \\ \hdashline
 &  & $10^{-9}$ & 0.018 (0.002) &  &  & $10^{-9}$ & 0.013 (0.001) &  &  & $10^{-9}$ & 0.015 (0.003) \\
 &  & $10^{-7}$ & 0.017 (0.002) &  &  & $10^{-7}$ & 0.014 (0.002) &  &  & $10^{-7}$ & 0.014 (0.002) \\
8 & 7 & $10^{-5}$ & 0.017 (0.002) & 16 & 7 & $10^{-5}$ & 0.013 (0.001) & 32 & 7 & $10^{-5}$ & 0.014 (0.002) \\
 &  & $10^{-3}$ & 0.026 (0.007) &  &  & $10^{-3}$ & 0.016 (0.002) &  &  & $10^{-3}$ & 0.014 (0.001) \\
 &  & $10^{-1}$ & 0.985 (0.022) &  &  & $10^{-1}$ & 0.985 (0.023) &  &  & $10^{-1}$ & 0.985 (0.022) \\
    \hline
    \end{tabular}
\end{table}

\clearpage

\subsection*{Scenario 2}
\leavevmode

\begin{align*}
    &\textbf{Model 1}  : \\
    &f^\circ(X_1,X_2, X_3, X_4, X_5, t) = \\
    & \quad \exp \{-15 \cdot \text{tanh}(|X_1| - 1/2)^2 \cdot (t+1) \} \cdot \sin \left\lbrace 50 \cdot (\text{gaussian}(\text{tanh}(|X_1|)) - 1/2)^2 \right\rbrace \\
    & \quad + \log \left \lbrace 1 + (\text{ReLU}(|X_1|) + \text{tanh}(|X_2|+ |X_3|) + \text{gaussian}(|X_4| + |X_5|))^2 \cdot (t+1) \right \rbrace, \\
    &\textbf{Model 2} : \\
    &f^\circ(X_1,X_2, X_3, X_4, X_5, t) = \\
    & \quad \exp \{-15 \cdot \text{tanh}(|X_1| - 1/2)^2 \cdot (t+1) \} \cdot \sin \left\lbrace 50 \cdot (\text{gaussian}(\text{tanh}(|X_1|)) - 1/2)^2 \right\rbrace, \\
    & \quad + \text{gaussian}\left\lbrace -50 \cdot (\text{tanh}(|X_2| - 1/2) )^2 \cdot (t+1) \right\rbrace \cdot \cos \left\lbrace (\exp (\text{gaussian}(|X_2|)) - 1/2)^2 \right\rbrace \\
    & \quad + \left\lbrace |X_1| + |X_2| \cdot (t+1) \right\rbrace / \left\lbrace 1 + |X_3 + X_4 + X_5| \cdot (t+1)^2 \right\rbrace, \\
    &\textbf{Model 3} : \\
    & f^\circ(X_1,X_2, X_3, X_4, X_5, t) = \\
    & \quad \exp \{-15 \cdot \text{tanh}(|X_1| - 1/2)^2 \cdot (t+1) \} \cdot \sin \left\lbrace 50 \cdot (\text{gaussian}(\text{tanh}(|X_1|)) - 1/2)^2 \right\rbrace \\
    & \quad + \text{gaussian}\left\lbrace -50 \cdot (\text{tanh}(|X_2| - 1/2) )^2 \cdot (t+1) \right\rbrace \cdot \cos \left\lbrace (\exp (\text{gaussian}(|X_2|)) - 1/2)^2 \right\rbrace \\
    & \quad - \text{tanh} \left\lbrace - 80 \cdot (\text{tanh}(|X_3| - 1/2))^2 \cdot (t + 1) \right\rbrace  \cdot \sin \left\lbrace (\text{ReLU}(\exp(|X_3|)) - 1/2)^2 \right\rbrace \\
    & \quad + \exp\left\lbrace -1 /10 \cdot (|X_1|^2 + |X_2|^2 + |X_3|^2) \right\rbrace\cdot \left\lbrace 1 + (|X_4|^2 + |X_5|^2) \cdot (t+1) \right\rbrace,
\end{align*}
where $\text{tanh}(x) = (e^x - e^{-x})/(e^x + e^{-x}), \ \text{gaussian}(x) = e^{- x^2}, \ \text
{ReLU}(x) = \max \{x, 0\} \;\text{for}\; x \in \mathbb{R}$.

\begin{figure}[h]
    \centering
    \includegraphics[width=\textwidth]{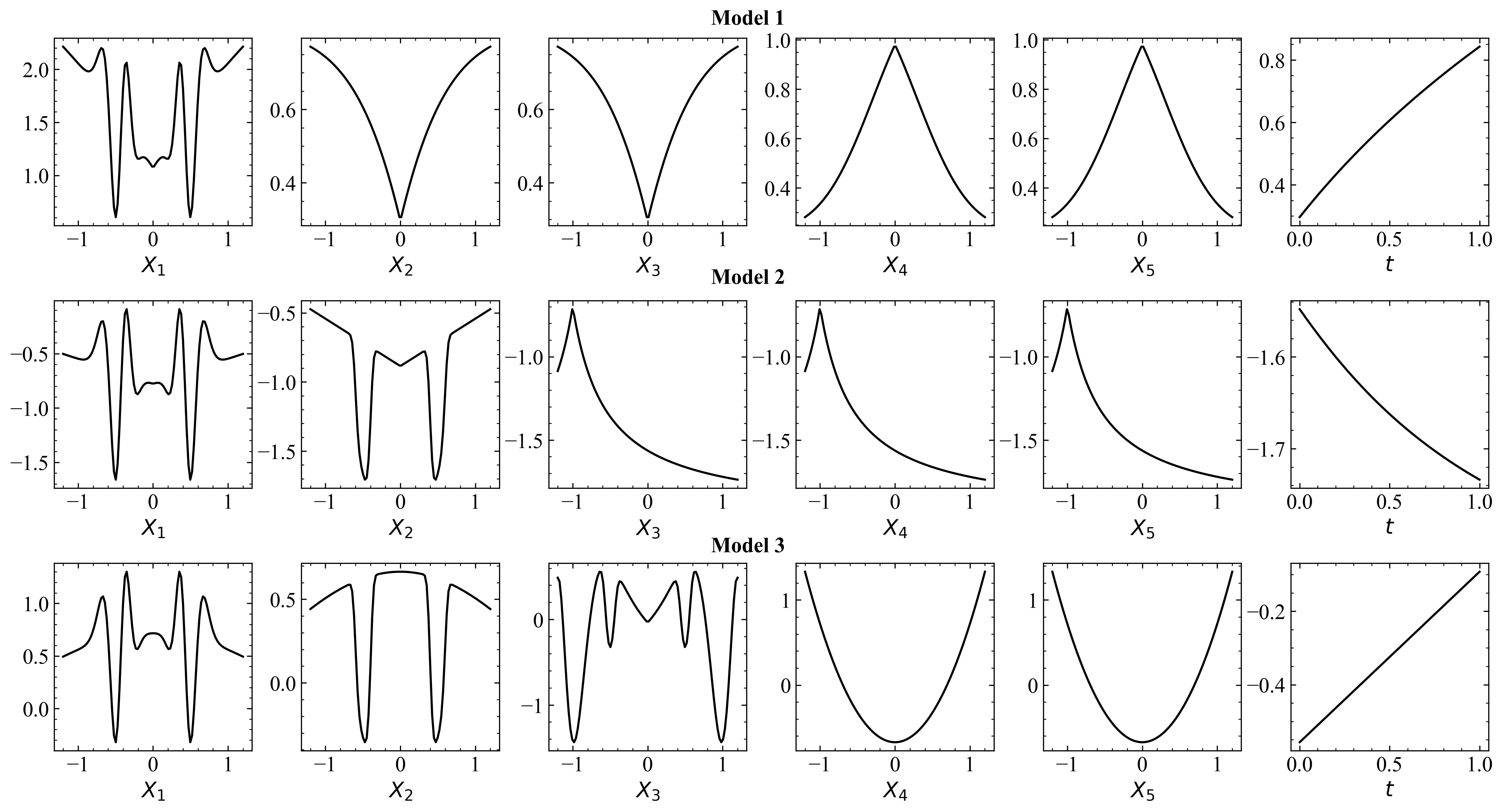}
    \caption{Plots of the functional form of $f(X_1, \ldots, X_5, t)$ for each of Model 1, 2, and 3 in Scenario 2, with $X_1, \ldots, X_5$, and $t$ as the horizontal axis.}
    \label{fig:Simulation function 3}
\end{figure}

\begin{figure}[h]
    \centering
    \includegraphics[width=\linewidth]{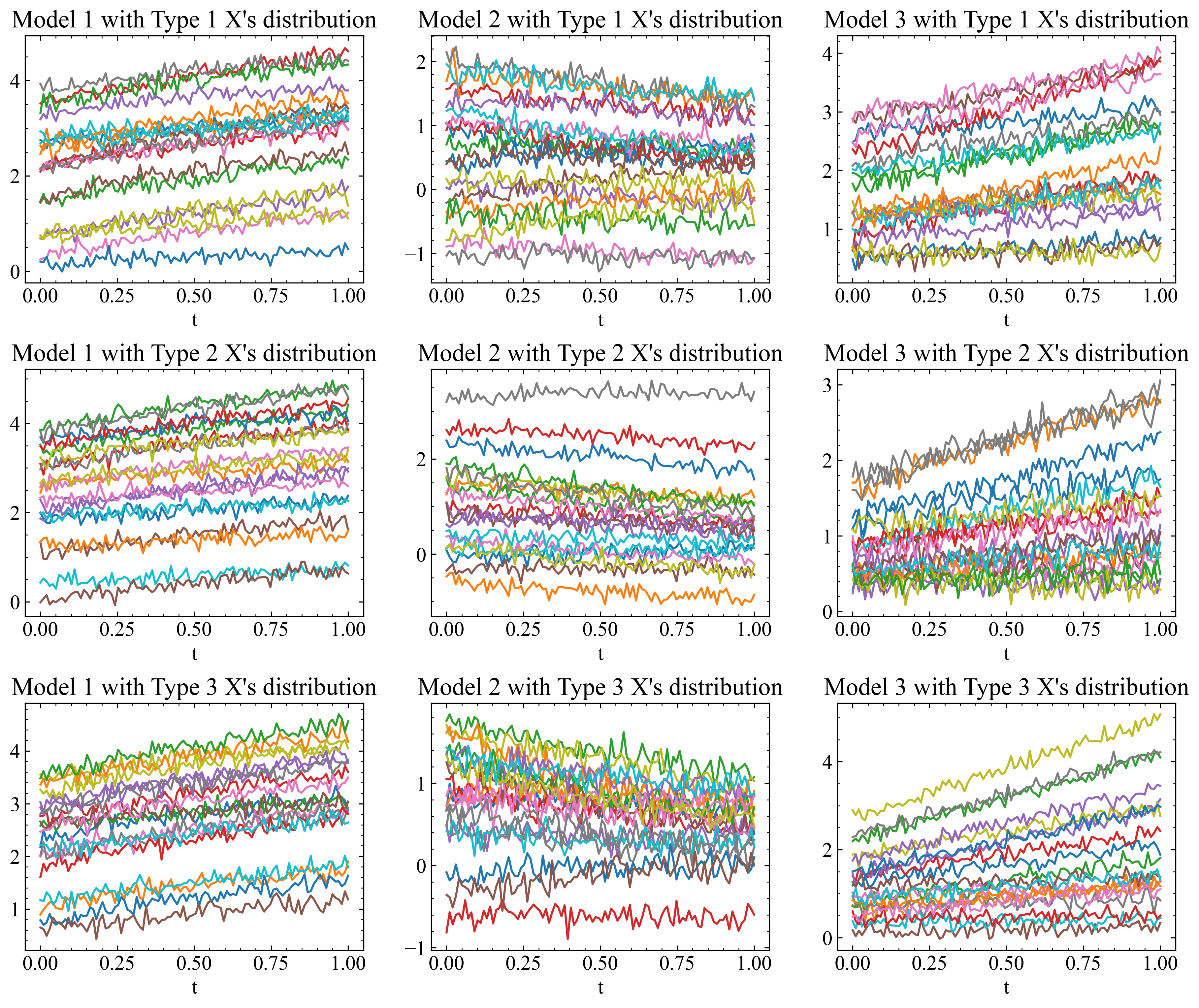}
    \caption{Sample curves of $Y(t)$ for the 9 settings which are nine different combinations of three models and three types of joint distributions of $(X_1, \ldots, X_5)$ in Scenario 2.}
    \label{fig:Simulation 3}
\end{figure}

\clearpage

\clearpage

\begin{table}[h]
    \centering
    \caption{Averages (and standard deviations) of the MISPEs from 50 replicates for all settings and methods in Scenario 2 ($N_{\text{train}} = 1000$).
    For the \textit{FOS-DNN}, we select $W = 32$, $L = 6$, and $\alpha = 10^{-5}$ for Model 1, $W = 32$, $L = 6$, and $\alpha = 10^{-5}$ for Model 2, $W = 32$, $L = 6$, and $\alpha = 10^{-5}$ for Model 3.
    For the \textit{FUA}, \textit{famm.nl}, and \textit{famm.ad}, we increase the number of basis functions by 10 compared to the default settings to address the high spatial inhomogeneity of the data. }
    \label{tab:12}
\begin{tabular}{|c|c|c|c|c|c|c|}
  \hline
   \textbf{X-type}  & \textbf{Model} & \textbf{FOS-DNN} &  \textbf{FUA} &  \textbf{famm.ad}  &  \textbf{fosr} \\
  \hline \hline 
   1  & 1 & 0.958 (0.110)  & \textbf{0.066 (0.088)}  & 0.428 (0.021) & 1.015 (0.022) \\
     & 2 & 0.701 (0.209)  & \textbf{0.088 (0.086)}  & 0.692 (0.029)  &  1.024 (0.031)  \\
     & 3 & 0.866 (0.115) & \textbf{0.193} (0.150) & 0.283 (0.017)  &  1.010 (0.028)  \\ \hline
   2  & 1 & 0.459 (0.029) & \textbf{0.072} (0.076) & 0.270 (0.012)  &  1.015 (0.021) \\
     & 2 & 0.540 (0.064) & \textbf{0.135 (0.128)} &  0.646 (0.031) & 0.991 (0.050) \\
     & 3 & 0.215 (0.022) & 0.186 (0.142)  & \textbf{0.126 (0.010)}  & 1.039 (0.067)  \\ \hline
   3    & 1 & 0.487 (0.020) & \textbf{0.068 (0.066)} & 0.270 (0.017) & 1.014 (0.020)  \\ 
       & 2 & 0.734 (0.087) & \textbf{0.177 (0.155)}  &  0.715 (0.033) & 1.022 (0.044)  \\ 
       & 3 & 0.300 (0.041) & \textbf{0.183 (0.135)} & 0.190 (0.014)  &  1.026 (0.063) \\ \hline 
\end{tabular} 
\end{table}

\clearpage

\begin{table}[t]
    \centering
    \caption{Averages (Standard deviations) of 3-fold Cross-Validation in Model 2 of Scenario 2 ($N_{\text{train}} = 1000$).}
    \label{tab:13}
    \begin{tabular}{|c|c|c|l||c|c|c|l||c|c|c|l|}
    \hline \hline
        W & L & $\alpha$ & mean (std) & W & L & $\alpha$ & mean (std) & W & L & $\alpha$ & mean (std) \\ \hline \hline
 &  & $10^{-9}$ & 1.050 (0.221) &  &  & $10^{-9}$ & 0.798 (0.043) &  &  & $10^{-9}$ & 0.589 (0.006) \\
   &   & $10^{-7}$ & 0.940 (0.257) &    &   & $10^{-7}$ & 0.738 (0.030) &    &   & $10^{-7}$ & 0.608 (0.003) \\
16   &  6 & $10^{-5}$ & 0.860 (0.121) &  32  &  6 & $10^{-5}$ & 0.767 (0.019) &  64  & 6  & $10^{-5}$ & 0.621 (0.045) \\
   &   & $10^{-3}$ & 0.520 (0.063) &    &   & $10^{-3}$ & 0.712 (0.057) &    &   & $10^{-3}$ & 0.555 (0.017) \\
   &   & $10^{-1}$ & 1.002 (0.075) &    &   & $10^{-1}$ & 1.002 (0.075) &    &   & $10^{-1}$ & 1.002 (0.075) \\ \hdashline
 &  & $10^{-9}$ & 0.978 (0.092) &  &  & $10^{-9}$ & 0.799 (0.017) &  &  & $10^{-9}$ & 0.527 (0.023) \\
   &   & $10^{-7}$ & 0.917 (0.020) &    &   & $10^{-7}$ & 0.786 (0.063) &    &   & $10^{-7}$ & 0.553 (0.035) \\
16   &  7 & $10^{-5}$ & 0.969 (0.210) &  32  & 7  & $10^{-5}$ & 0.677 (0.010) &  64  & 7  & $10^{-5}$ & 0.527 (0.029) \\
   &   & $10^{-3}$ & 0.483 (0.013) &    &   & $10^{-3}$ & 0.643 (0.040) &    &   & $10^{-3}$ & 0.543 (0.033) \\
   &   & $10^{-1}$ & 1.002 (0.075) &    &   & $10^{-1}$ & 1.002 (0.075) &    &   & $10^{-1}$ & 1.002 (0.075) \\ \hdashline
 &  & $10^{-9}$ & 0.933 (0.086) &  &  & $10^{-9}$ & 0.686 (0.046) &  &  & $10^{-9}$ & 0.532 (0.045) \\
   &   & $10^{-7}$ & 1.058 (0.231) &    &   & $10^{-7}$ & 0.741 (0.033) &    &   & $10^{-7}$ & 0.532 (0.027) \\
16   &  8 & $10^{-5}$ & 0.857 (0.021) &  32  & 8  & $10^{-5}$ & 0.688 (0.008) &  64  & 8  & $10^{-5}$ & 0.544 (0.027) \\
   &   & $10^{-3}$ & 0.583 (0.063) &    &   & $10^{-3}$ & 0.712 (0.103) &    &   & $10^{-3}$ & 0.538 (0.041) \\
   &   & $10^{-1}$ & 1.002 (0.075) &    &   & $10^{-1}$ & 1.002 (0.075) &    &   & $10^{-1}$ & 1.002 (0.075) \\
    \hline
    \end{tabular}
\end{table}

\begin{table}[t]
    \centering
    \caption{Averages (Standard deviations) of 3-fold Cross-Validation in Model 2 of Scenario 2 ($N_{\text{train}} = 5000$).}
    \label{tab:14}
    \begin{tabular}{|c|c|c|l||c|c|c|l||c|c|c|l|}
    \hline \hline
        W & L & $\alpha$ & mean (std) & W & L & $\alpha$ & mean (std) & W & L & $\alpha$ & mean (std) \\ \hline \hline
 &  & $10^{-9}$ & 0.041 (0.032) &  &  & $10^{-9}$ & 0.279 (0.080) &  &  & $10^{-9}$ & 0.339 (0.042) \\
 &  & $10^{-7}$ & 0.034 (0.018) &  &  & $10^{-7}$ & 0.378 (0.077) &  &  & $10^{-7}$ & 0.339 (0.066) \\
16 & 6 & $10^{-5}$ & 0.021 (0.004) & 32 & 6 & $10^{-5}$ & 0.225 (0.168) & 64 & 6 & $10^{-5}$ & 0.310 (0.022) \\
 &  & $10^{-3}$ & 0.055 (0.035) &  &  & $10^{-3}$ & 0.021 (0.001) &  &  & $10^{-3}$ & 0.020 (0.001) \\
 &  & $10^{-1}$ & 1.009 (0.020) &  &  & $10^{-1}$ & 1.009 (0.019) &  &  & $10^{-1}$ & 1.009 (0.019) \\ \hdashline
 &  & $10^{-9}$ & 0.190 (0.075) &  &  & $10^{-9}$ & 0.439 (0.061) &  &  & $10^{-9}$ & 0.323 (0.039) \\
 &  & $10^{-7}$ & 0.107 (0.065) &  &  & $10^{-7}$ & 0.278 (0.045) &  &  & $10^{-7}$ & 0.358 (0.015) \\
16 & 7 & $10^{-5}$ & 0.014 (0.001) & 32 & 7 & $10^{-5}$ & 0.325 (0.040) & 64 & 7 & $10^{-5}$ & 0.272 (0.049) \\
 &  & $10^{-3}$ & 0.023 (0.005) &  &  & $10^{-3}$ & 0.021 (0.004) &  &  & $10^{-3}$ & 0.018 (0.000) \\
 &  & $10^{-1}$ & 1.009 (0.019) &  &  & $10^{-1}$ & 1.009 (0.019) &  &  & $10^{-1}$ & 1.009 (0.020) \\ \hdashline
 &  & $10^{-9}$ & 0.176 (0.097) &  &  & $10^{-9}$ & 0.425 (0.051) &  &  & $10^{-9}$ & 0.316 (0.020) \\
 &  & $10^{-7}$ & 0.202 (0.025) &  &  & $10^{-7}$ & 0.362 (0.125) &  &  & $10^{-7}$ & 0.358 (0.036) \\
16 & 8 & $10^{-5}$ & 0.102 (0.107) & 32 & 8 & $10^{-5}$ & 0.389 (0.026) & 64 & 8 & $10^{-5}$ & 0.322 (0.021) \\
 &  & $10^{-3}$ & 0.023 (0.004) &  &  & $10^{-3}$ & 0.017 (0.000) &  &  & $10^{-3}$ & 0.018 (0.003) \\
 &  & $10^{-1}$ & 1.009 (0.019) &  &  & $10^{-1}$ & 1.009 (0.020) &  &  & $10^{-1}$ & 1.009 (0.020) \\
    \hline
    \end{tabular}
\end{table}

\clearpage

\subsection*{Scenario 3}
\leavevmode
\begin{align*}
    &\textbf{Model 1} : \\
    &f^\circ(X_1,\ldots, X_{10}, t) = \\
    & \exp \{-15 \cdot \text{tanh}(|X_1| - 1/2)^2 (t + 1)\} \cdot \sin \{50 \cdot (\text{gaussian}(\text{tanh}(|X_1|))- 1/2)^2\} \\
    &+ \left\lbrace\log (1 + (\text{ReLU}(|X_1| + \cdots + |X_4|) + \text{tanh}(|X_5| + \cdots + |X_7|) \right.\\
    & \hspace{4cm} + \left.\text{gaussian}(|X_8| + \cdots + |X_{10}|))^2 \right\rbrace \cdot (t + 1), \\
    &\textbf{Model 2} : \\
    &f^\circ(X_1,\ldots, X_{10}, t) = \\
    & \exp \{- 15 \cdot \text{tanh}(|X_1| - 1/2)^2 \cdot (t+1) \} \cdot \sin \{50 \cdot (\text{gaussian}(\text{tanh}(|X_1|))- 1/2)^2 \} \\
    &+ \text{gaussian}\{-50 \cdot \text{tanh}(|X_2|-1/2)^2 \cdot (t+1)\} \cdot \cos \{ (\exp(\text{gaussian}(X_2)) - 1/2)^2\} \\
    &+ \text{logit} \{-80 \cdot \text{tanh}(|X_3|-1/2)^2 \cdot (t+1)\} \cdot \sin \{5 \cdot (\text{ReLU}(\exp(|X_3|))-1/2)^2\} \\
    & + \left\lbrace |X_1|+ \cdots + |X_3| + (|X_4| + |X_5|) \cdot (t+1) \right\rbrace / \left\lbrace 1 + |X_6 + \cdots + X_{10}| \cdot (t+1)^2 \right\rbrace, \\
    &\textbf{Model 3} : \\
    &f^\circ(X_1,\ldots, X_{10}, t) = \\
    & \exp \{- 15 \cdot \text{tanh}(|X_1| - 1/2)^2 \cdot (t+1) \} \cdot \sin \{50 \cdot (\text{gaussian}(\text{tanh}(|X_1|)- 1/2))^2 \} \\
    &+ \text{gaussian}\{-50 \cdot \text{tanh}(|X_2|-1/2)^2 \cdot (t+1)\} \cdot \cos \{ (\exp(\text{gaussian}(X_2)) - 1/2)^2\} \\
    &+ \text{logit} \{-80 \cdot \text{tanh}(|X_3|-1/2)^2 \cdot (t+1)\} \cdot \sin \{5 \cdot (\text{ReLU}(\exp(|X_3|))-1/2)^2\} \\
    &- \text{gaussian}\{-50 \cdot \text{tanh}(|X_4|-1/2)^2 \cdot (t+1)\} \cdot \cos \{ (\exp(\text{gaussian}(X_4)) - 1/2)^2\} \\
    &- \exp \{- 30 \cdot \text{tanh}(|X_5| - 1/2)^2 \cdot (t+1) \} \cdot \sin \{50 \cdot (\text{gaussian}(\text{tanh}(|X_5|)- 1/2))^2 \} \\
    &+ \exp\left\lbrace -1 /10 \cdot (|X_1|^2 + \cdots + |X_5|^2) \right\rbrace\cdot \left\lbrace 1 + (|X_6|^2 + \cdots + |X_{10}|^2) \cdot (t+1) \right\rbrace.
\end{align*}

\begin{figure}[h]
    \centering
    \includegraphics[width=\textwidth]{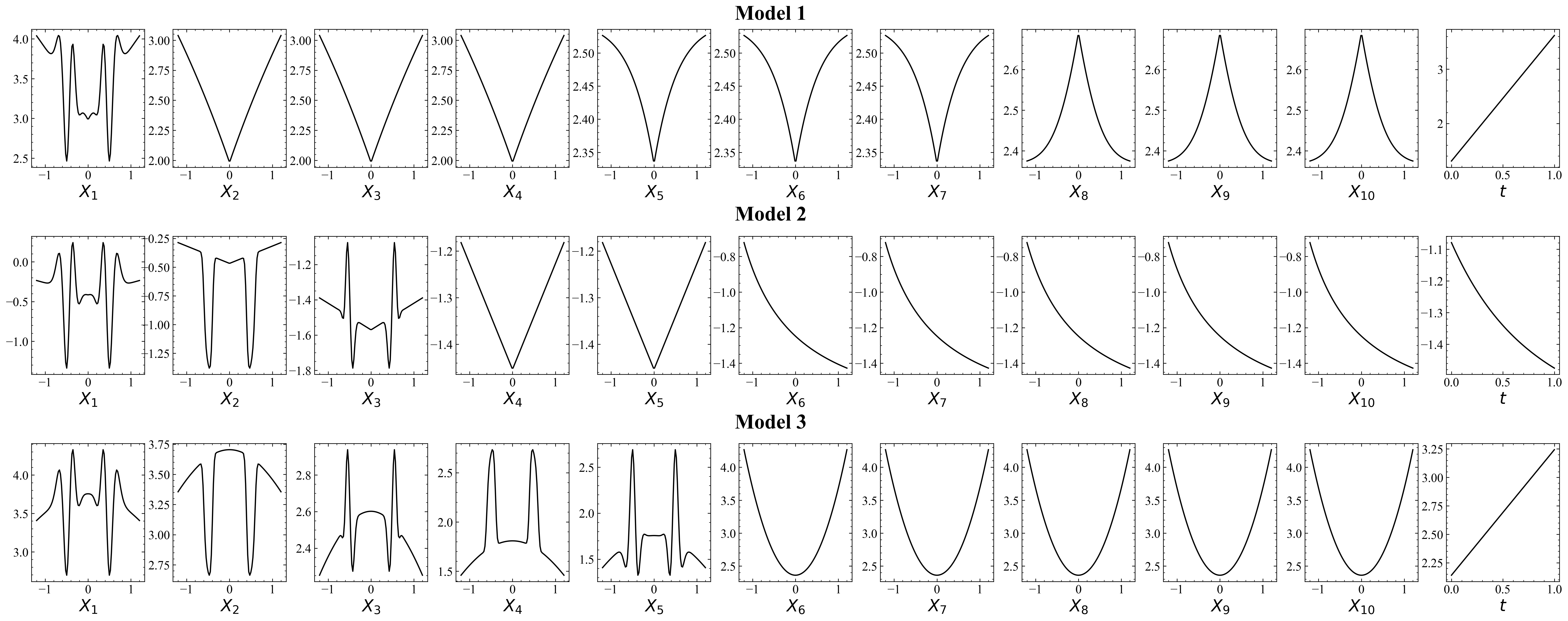}
    \caption{Plots of the functional form of $f(X_1, \ldots, X_{10}, t)$ for each of Model 1, 2, and 3 in Scenario 3, with $X_1, \ldots, X_{10}$, and $t$ as the horizontal axis.}
    \label{fig:Simulation function 4}
\end{figure}

\begin{figure}[h]
    \centering
    \includegraphics[width=\linewidth]{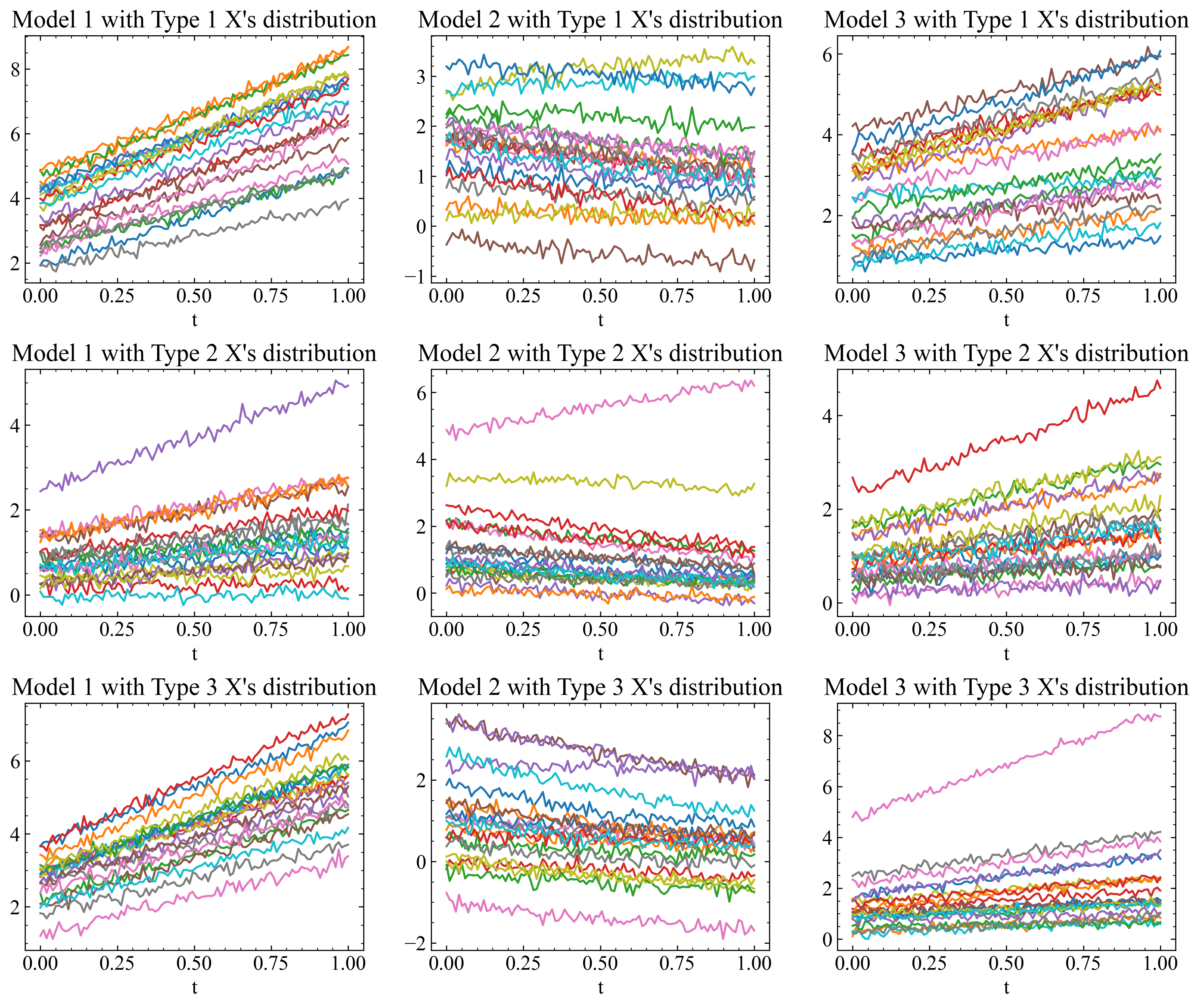}
    \caption{Sample curves of $Y(t)$ for the 9 settings which are nine different combinations of three models and three types of joint distributions of $(X_1, \ldots, X_{10})$ in Scenario 3.}
    \label{fig:Simulation 4}
\end{figure}

\clearpage

\begin{table}[h]
    \centering
    \caption{Averages (and standard deviations) of the MISPEs from 20 replicates for all settings and methods in Scenario 3 ($N_{\text{train}} = 5000$). 
    For the \textit{FOS-DNN}, we select $W = 32$, $L = 7$, and $\alpha = 10^{-3}$ for Model 1, $W = 32$, $L = 7$, and $\alpha = 10^{-3}$ for Model 2, $W = 32$, $L = 7$, and $\alpha = 10^{-3}$ for Model 3. 
    For the \textit{FUA}, \textit{famm.nl}, and \textit{famm.ad}, we increase the number of basis functions by 10 compared to the default settings to address the high spatial inhomogeneity of the data. }
    \label{tab:15}
\begin{tabular}{|c|c|c|c|c|c|c|}
  \hline
   \textbf{X-type}  & \textbf{Model} & \textbf{FOS-DNN} &  \textbf{FUA} &  \textbf{famm.ad}  &  \textbf{fosr} \\
  \hline \hline 
   1  & 1 &  0.174 (0.202)  & \textbf{0.062 (0.003)}  & 0.308 (0.013)  &  1.007 (0.043) \\
     & 2 &  \textbf{0.137 (0.137)}  & 0.191 (0.022)  & 0.789 (0.037)  &  1.027 (0.058)    \\
     & 3 & 0.515 (0.074) &  0.318 (0.024) & \textbf{0.279 (0.015)}  &  1.012 (0.033)  \\ \hline
   2  & 1 &  \textbf{0.093 (0.082)}  &  0.102 (0.041) & 0.249 (0.208) & 1.028 (0.047)   \\
     & 2 &  \textbf{0.105 (0.042)}  &  0.187 (0.011) & 0.778 (0.050) & 1.006 (0.073) \\
     & 3 & \textbf{0.058 (0.004)} &  0.103 (0.009)  & 0.108 (0.009)   &  1.037 (0.062) \\ \hline
   3    & 1  & \textbf{0.032 (0.042)}  & 0.074 (0.016)  & 0.128 (0.007)  & 1.028 (0.038)  \\ 
       & 2 & \textbf{0.268 (0.070)}  & 0.281 (0.031)  & 0.770 (0.040)  & 1.004 (0.052)  \\ 
       & 3 & \textbf{0.068 (0.004)} & 0.089 (0.009) &  0.182 (0.010)  & 1.012 (0.053)  \\ \hline 
\end{tabular} 
\end{table}

\clearpage

\begin{table}[t]
    \centering
    \caption{Averages (Standard deviations) of 3-fold Cross-Validation in Model 2 of Scenario 3 ($N_{\text{train}} = 5000$).}
    \label{tab:16}
    \begin{tabular}{|c|c|c|l||c|c|c|l||c|c|c|l|}
    \hline \hline
        W & L & $\alpha$ & mean (std) & W & L & $\alpha$ & mean (std) & W & L & $\alpha$ & mean (std) \\ \hline \hline
 &  & $10^{-9}$ & 0.232 (0.018) &  &  & $10^{-9}$ & 0.485 (0.012) &  &  & $10^{-9}$ & 0.350 (0.010) \\
 &  & $10^{-7}$ & 0.210 (0.044) &  &  & $10^{-7}$ & 0.552 (0.005) &  &  & $10^{-7}$ & 0.339 (0.002) \\
16 & 6 & $10^{-5}$ & 0.208 (0.037) & 32 & 6 & $10^{-5}$ & 0.469 (0.012) & 64 & 6 & $10^{-5}$ & 0.311 (0.006) \\
 &  & $10^{-3}$ & 0.129 (0.016) &  &  & $10^{-3}$ & 0.214 (0.040) &  &  & $10^{-3}$ & 0.257 (0.025) \\
 &  & $10^{-1}$ & 1.039 (0.039) &  &  & $10^{-1}$ & 1.040 (0.039) &  &  & $10^{-1}$ & 1.040 (0.039) \\ \hdashline
 &  & $10^{-9}$ & 0.289 (0.026) &  &  & $10^{-9}$ & 0.461 (0.014) &  &  & $10^{-9}$ & 0.330 (0.001) \\
 &  & $10^{-7}$ & 0.269 (0.014) &  &  & $10^{-7}$ & 0.481 (0.027) &  &  & $10^{-7}$ & 0.327 (0.007) \\
16 & 7 & $10^{-5}$ & 0.225 (0.046) & 32 & 7 & $10^{-5}$ & 0.445 (0.016) & 64 & 7 & $10^{-5}$ & 0.311 (0.005) \\
 &  & $10^{-3}$ & 0.112 (0.003) &  &  & $10^{-3}$ & 0.206 (0.012) &  &  & $10^{-3}$ & 0.252 (0.020) \\
 &  & $10^{-1}$ & 1.039 (0.038) &  &  & $10^{-1}$ & 1.039 (0.039) &  &  & $10^{-1}$ & 1.039 (0.038) \\ \hdashline
 &  & $10^{-9}$ & 0.243 (0.019) &  &  & $10^{-9}$ & 0.458 (0.007) &  &  & $10^{-9}$ & 0.309 (0.017) \\
 &  & $10^{-7}$ & 0.269 (0.014) &  &  & $10^{-7}$ & 0.445 (0.007) &  &  & $10^{-7}$ & 0.320 (0.011) \\
16 & 8 & $10^{-5}$ & 0.283 (0.028) & 32 & 8 & $10^{-5}$ & 0.466 (0.017) & 64 & 8 & $10^{-5}$ & 0.303 (0.009) \\
 &  & $10^{-3}$ & 0.117 (0.004) &  &  & $10^{-3}$ & 0.262 (0.032) &  &  & $10^{-3}$ & 0.293 (0.016) \\
 &  & $10^{-1}$ & 1.039 (0.039) &  &  & $10^{-1}$ & 1.039 (0.038) &  &  & $10^{-1}$ & 1.039 (0.038) \\
    \hline
    \end{tabular}
\end{table}

\begin{table}[t]
    \centering
    \caption{Averages (Standard deviations) of 3-fold Cross-Validation in Model 2 of Scenario 3 ($N_{\text{train}} = 10000$).}
    \label{tab:17}
    \begin{tabular}{|c|c|c|l||c|c|c|l||c|c|c|l|}
    \hline \hline
        W & L & $\alpha$ & mean (std) & W & L & $\alpha$ & mean (std) & W & L & $\alpha$ & mean (std) \\ \hline \hline
 &  & $10^{-9}$ & 0.139 (0.023) &  &  & $10^{-9}$ & 0.222 (0.059) &  &  & $10^{-9}$ & 0.355 (0.010) \\
 &  & $10^{-7}$ & 0.137 (0.027) &  &  & $10^{-7}$ & 0.271 (0.031) &  &  & $10^{-7}$ & 0.338 (0.030) \\
16 & 6 & $10^{-5}$ & 0.099 (0.058) & 32 & 6 & $10^{-5}$ & 0.268 (0.057) & 64 & 6 & $10^{-5}$ & 0.310 (0.008) \\
 &  & $10^{-3}$ & 0.109 (0.025) &  &  & $10^{-3}$ & 0.087 (0.039) &  &  & $10^{-3}$ & 0.156 (0.011) \\
 &  & $10^{-1}$ & 1.020 (0.055) &  &  & $10^{-1}$ & 1.019 (0.054) &  &  & $10^{-1}$ & 1.019 (0.054) \\ \hdashline
 &  & $10^{-9}$ & 0.106 (0.005) &  &  & $10^{-9}$ & 0.356 (0.015) &  &  & $10^{-9}$ & 0.308 (0.006) \\
 &  & $10^{-7}$ & 0.115 (0.003) &  &  & $10^{-7}$ & 0.349 (0.003) &  &  & $10^{-7}$ & 0.308 (0.008) \\
16 & 7 & $10^{-5}$ & 0.136 (0.032) & 32 & 7 & $10^{-5}$ & 0.294 (0.037) & 64 & 7 & $10^{-5}$ & 0.289 (0.028) \\
 &  & $10^{-3}$ & 0.100 (0.003) &  &  & $10^{-3}$ & 0.068 (0.035) &  &  & $10^{-3}$ & 0.157 (0.011) \\
 &  & $10^{-1}$ & 1.020 (0.055) &  &  & $10^{-1}$ & 1.019 (0.054) &  &  & $10^{-1}$ & 1.020 (0.054) \\ \hdashline
 &  & $10^{-9}$ & 0.117 (0.010) &  &  & $10^{-9}$ & 0.310 (0.055) &  &  & $10^{-9}$ & 0.299 (0.010) \\
 &  & $10^{-7}$ & 0.149 (0.028) &  &  & $10^{-7}$ & 0.337 (0.009) &  &  & $10^{-7}$ & 0.299 (0.017) \\
16 & 8 & $10^{-5}$ & 0.121 (0.003) & 32 & 8 & $10^{-5}$ & 0.308 (0.039) & 64 & 8 & $10^{-5}$ & 0.282 (0.017) \\
 &  & $10^{-3}$ & 0.097 (0.007) &  &  & $10^{-3}$ & 0.027 (0.002) &  &  & $10^{-3}$ & 0.186 (0.001) \\
 &  & $10^{-1}$ & 1.019 (0.054) &  &  & $10^{-1}$ & 1.019 (0.053) &  &  & $10^{-1}$ & 1.064 (0.038) \\
    \hline
    \end{tabular}
\end{table}

\clearpage

\subsection*{GRF data}
\leavevmode
\begin{table}[h]
    \centering
    \caption{Averages (Standard deviations) of prediction errors from 10 replicates for GRF data.}
    \label{tab:18}
    \begin{tabular}{|c|c|c|l||c|c|c|l||c|c|c|l|}
    \hline \hline
        W & L & $\alpha$ & mean (std) & W & L & $\alpha$ & mean (std) \\ \hline \hline
 &  & $10^{-9}$ & 0.026 (0.000) &  &  & $10^{-9}$ & 0.024 (0.000)  \\
32 & 7 & $10^{-7}$ & 0.028 (0.000) & 64 & 7 & $10^{-7}$ & 0.024 (0.000)  \\
 &  & $10^{-5}$ & 0.027 (0.000) &  &  & $10^{-5}$ & 0.023 (0.000) \\ \hdashline
 &  & $10^{-9}$ & 0.027 (0.000) &  &  & $10^{-9}$ & 0.024 (0.000) \\
32 & 8 & $10^{-7}$ & 0.027 (0.000) & 64 & 8 & $10^{-7}$ & 0.023 (0.000)  \\
 &  & $10^{-5}$ & 0.028 (0.000) &  &  & $10^{-5}$ & 0.025 (0.000)  \\ \hdashline
 &  & $10^{-9}$ & 0.026 (0.000) &  &  & $10^{-9}$ & 0.024 (0.000)  \\
32 & 9 & $10^{-7}$ & 0.025 (0.000) & 64 & 9 & $10^{-7}$ & 0.025 (0.000) \\
 &  & $10^{-5}$ & 0.029 (0.000) &  &  & $10^{-5}$ & 0.024 (0.000)  \\
    \hline
    \end{tabular}
\end{table}

\begin{table}[h]
    \centering
    \caption{Averages (Standard deviations) of FUA prediction errors according to the number of basis functions.}
    \label{tab:19}
    \begin{tabular}{|c|c|c|c|c|} \hline \hline
     \# of basis & $10$ & $ 15 $ & $ 20 $ & $ 30 $ \\ \hline \hline
     FUA   & 0.033 (0.001) & 0.030 (0.001) & 0.028 (0.001) & 0.027 (0.001) \\
     \hline
    \end{tabular}
\end{table}

\begin{sidewaystable}[htbp]
\centering
\caption{Summary of predictor variables and their descriptions}
\label{tab:variables_summary}
\resizebox{0.95\textwidth}{!}{
\begin{tabular}{lll}
\hline
\textbf{Variable Name} & \textbf{Description} & \textbf{Value} \\ \hline
\text{Age} & Age (years) & $14.51 \pm 2.20$ \\
\text{Height} & Height (cm) & $165.57 \pm 8.82$ \\
\text{Weight} & Weight (kg) & $56.33 \pm 9.06$ \\
\text{soccer\_use\_foot} & Whether the tested foot is used to kick a soccer ball (Yes=1, No=0, Neutral=0.5)  & 1 : 1224, 0 : 1208, 0.5 : 33 \\
\text{jump\_use\_foot} & Whether the tested foot is used to jump upwards (Yes=1, No=0, Neutral=0.5) & 1 : 1186, 0 : 1213, 0.5 : 66 \\
\text{fall\_use\_foot} & Whether the tested foot is used as a reflex to prevent falling (Yes=1, No=0, Neutral=0.5) & 1 : 1215, 0 : 1207, 0.5 : 43 \\
\text{landing\_use\_foot} & Whether the tested foot feels stable when landing (Yes=1, No=0, Neutral=0.5) & 1 : 1214, 0 : 1168, 0.5 : 83 \\
\text{FL} & Foot length (cm) & 24.64 $\pm$ 1.61 \\
\text{FW} & Foot width (cm) & 9.71 $\pm$ 0.71 \\
\text{Sports\_career} & Sports career (years) & 5.38 $\pm$ 3.07 \\
\text{ACL\_counts} & Number of ACL injuries & 1 : 35, 0 : 2430 \\
\text{Injure\_exp} & Whether the tested foot has been injured (Yes=1, No=0) & 1 : 796, 0 : 1669 \\
\text{Contact} & Whether the injury was caused by contact (Yes=1, No=0) & 1 : 111, 0 : 2354 \\
\text{NonContact} & Whether the injury was caused by non-contact (Yes=1, No=0) & 1 : 185, 0 : 2280 \\
\text{Sprain} & Whether the injury was caused by a sprain (Yes=1, No=0) & 1 : 471, 0 : 1994 \\
\text{Since\_injure\_3m} & Whether the injury occurred within 3 months before the experiment (Yes=1, No=0) & 1 : 152, 0 : 2313 \\
\text{Since\_injure\_6m} & Whether the injury occurred 3-6 months before the experiment (Yes=1, No=0) & 1 : 136, 0 : 2329 \\
\text{Since\_injure\_1y} & Whether the injury occurred 6 months-1 year before the experiment (Yes=1, No=0) & 1 : 275, 0 : 2190 \\
\text{Since\_injure\_2y} & Whether the injury occurred 1-2 years before the experiment (Yes=1, No=0) & 1 : 107, 0 : 2358 \\
\text{Since\_injure\_over2y} & Whether the injury occurred more than 2 years before the experiment (Yes=1, No=0) & 1 : 120, 0 : 2345 \\
\text{Recover\_1w} & Whether recovery occurred within 1 week (Yes=1, No=0) & 1 : 338, 0 : 2127 \\
\text{Recover\_1m} & Whether recovery occurred within 1 week-1 month (Yes=1, No=0) & 1 : 389, 0 : 2076 \\
\text{Recover\_over1m} & Whether recovery took more than 1 month (Yes=1, No=0) & 1 : 54, 0 : 2411 \\
\text{Recover\_3m} & Whether recovery occurred within 1-3 months (Yes=1, No=0) & 1 : 4, 0 : 2461 \\
\text{Recover\_1y} & Whether recovery occurred within 3 months-1 year (Yes=1, No=0) & 1 : 5, 0 : 2460 \\ \hline
\end{tabular}}
\end{sidewaystable}

\newpage

\section{Note on Theorem~3 of \cite{luo2023nonlinear}} \label{sec: Appendix D}

This appendix aims to clarify certain theoretical aspects of the proof of Theorem 3 presented in \cite{luo2023nonlinear}.
In Theorem~3 of \cite{luo2023nonlinear}, 
they claim that the convergence rate of the FUA estimator is the parametric rate $O_p (n^{- 1/2})$.
However, in the context of nonparametric estimation, this convergence rate is unrealistic without any strong assumptions.
Upon careful examination, we identified discrepancies that may affect the conclusions drawn from this proof.
In the proof of Theorem 3 in \cite{luo2023nonlinear}, specifically in Web Appendix A.2, it is claimed that the following expression holds by the classical central limit theorem:
\[
    \frac{1}{n} \sum_{l=1}^n \left\lbrack \|\mathcal{H}_t(\boldsymbol{X}_l) - \hat{\mathcal{H}}_t(\boldsymbol{X}_l)\|_{L^2}^2 - m_{\mathcal{H} - \hat{\mathcal{H}}} \right\rbrack = O_p (n^{- 1/2}),
\]
where $m_{\mathcal{H} - \hat{\mathcal{H}}} = E_{X_{\text{new}}} \left\lbrack \|\mathcal{H}_t (\boldsymbol{X}_{\text{new}}) - \hat{\mathcal{H}}_t (\boldsymbol{X}_{\text{new}}) \|_{L^2}^2 \right\rbrack$. 
However, since the estimator $\hat{\mathcal{H}}$ depends on the training data $\{\boldsymbol{X}_l , Y_l\}_{l=1}^n$, 
the classical central limit theorem cannot be applied. 
Therefore, the convergence rate of Theorem 3 in \cite{luo2023nonlinear} is not achieved without additional strong assumptions.

\end{appendix}

\bibliographystyle{imsart-nameyear}
\bibliography{paper-ref}

\end{document}